\documentclass[aip,jmp,11pt,a4paper]{revtex4-1}
\usepackage{amsmath}
\usepackage{amsfonts}
\usepackage{amssymb}
\usepackage{graphicx}

\providecommand{\U}[1]{\protect\rule{.1in}{.1in}}
\newtheorem{theorem}{Theorem}
\newtheorem{acknowledgement}{Acknowledgement}

\newtheorem{corollary}{Corollary}

\newtheorem{definition}{Definition}

\newtheorem{lemma}{Lemma}
\newtheorem{notation}{Notation}

\newtheorem{proposition}{Proposition}
\newtheorem{remark}{Remark}

\newenvironment{proof}[1][Proof]{\noindent\textbf{#1.} }{\ \rule{0.5em}{0.5em}}
\input{tcilatex}
\begin{document}

\title{Local quasi hidden variable modelling and violations of Bell-type
inequalities by a multipartite quantum state}
\author{Elena R. Loubenets}
\affiliation{Applied Mathematics Department, Moscow State Institute of Electronics and
Mathematics, Moscow 109028, Russia}

\begin{abstract}
We introduce for a general correlation scenario a new simulation model, 
\emph{a} \emph{local quasi hidden variable (LqHV) model, }where\emph{\ }%
locality and the measure-theoretic structure inherent to an LHV model are
preserved but positivity of a simulation measure is dropped. We specify a
necessary and sufficient condition for LqHV modelling and, based on this,
prove that every quantum correlation scenario admits an LqHV simulation. Via
the LqHV approach, we construct analogs of Bell-type inequalities for an $N$%
-partite quantum state and find a new analytical upper bound on the maximal
violation by an $N$-partite quantum state of $S_{1}\times \cdots \times S_{N}
$-setting Bell-type inequalities - either on correlation functions or on
joint probabilities and for outcomes of an arbitrary spectral type, discrete
or continuous. This general analytical upper bound is expressed in terms of
the new state dilation characteristics introduced in the present paper and
not only traces quantum states admitting an $S_{1}\times \cdots \times S_{N}$%
-setting LHV description but also leads to the new exact numerical upper
estimates on the maximal Bell violations for concrete $N$-partite quantum
states used in quantum information processing and for an arbitrary $N$%
-partite quantum state. We, in particular, prove that violation by an $N$%
-partite quantum state of an arbitrary Bell-type inequality (either on
correlation functions or on joint probabilities) for $S$ settings per site
cannot exceed $(2S-1)^{N-1}$ even in case of an infinite dimensional quantum
state and infinitely many outcomes.
\end{abstract}

\maketitle
\tableofcontents

\section{Introduction}

The seminal papers of Einstein, Podolsky and Rosen \cite{EPR} (EPR) and Bell 
\cite{Bell, 1} are still ones of most cited in quantum information. In Ref. 
\cite{EPR}, Einstein, Podolsky and Rosen argued that \emph{locality} of
measurements performed by spatially separated parties on perfectly
correlated quantum events implies the "simultaneous reality - and thus
definite values" of physical quantities described by noncommuting quantum
observables. Based on this argument contradicting, however, the quantum
formalism and referred to as the EPR paradox, Einstein, Podolsky and Rosen
expressed in Ref. \cite{EPR} their belief on a possibility of a \emph{hidden
variable} account of quantum measurements.

Analyzing this EPR belief in 1964 - 1966, Bell explicitly constructed \cite%
{Bell} the hidden variable (HV) model reproducing the statistical properties
of all quantum observables for a qubit but, however, proved \cite{1} that,
for\emph{\ }bipartite\emph{\ }measurements on a two-qubit system in the
singlet state, a\emph{\ local} hidden variable description (LHV) disagrees
with the statistical predictions of quantum theory. Based on these results,
Bell concluded \cite{Bell} that the EPR paradox should be resolved
specifically via violation of \emph{locality} under bipartite quantum
measurements and that "...non-locality is deeply rooted in quantum mechanics
itself and will persist in any completion".

Ever since 1964, the conceptual and mathematical aspects of the
probabilistic description of multipartite quantum measurements have been
analyzed in a plenty of papers, see, for example, articles \cite{kochen,
tsirelson, tsirelson1, fine, werner, Popescu, loubenets3} and references
therein. Nevertheless, as it has been recently noted by Gisin \cite{gisin},
in this field there are still "many questions, a few answers".

It was, for example, proved by Werner \cite{werner} that there exist finite
dimensional nonseparable bipartite quantum states admitting an LHV
description under all projective bipartite quantum measurements with an
arbitrary number of measurement settings at each site. It was also shown in
Refs. \cite{terhal, loubenets2, loubenets2a, loubenets2b, loubenets3} that
some nonseparable bipartite quantum states admit an LHV description only
under correlation scenarios with specific numbers of measurements at $N$
sites. However, until now it is not still known what state parameter
quantitatively determines violation by an $N$-partite quantum state of
Bell-type inequalities \cite{fr} - constraints specifying scenarios
admitting an LHV description and named after the seminal result \cite{1} of
Bell.

Nowadays, it is also clear \cite{loubenets3} that though multipartite
quantum measurements do not need to be local in the sense of Bell, they are,
however, local in the sense meant by Einstein \emph{et al} in Ref. \cite{EPR}%
. The difference between the general nonsignaling condition, the EPR
locality and Bell's locality is analyzed in Ref. \cite{loubenets3}. Thus,
the term "a nonlocal quantum state" widespread in quantum information means
now only that this state does not admit an LHV\ description and, therefore,
violates some Bell-type inequality.

This takes us back to the EPR locality argument \cite{EPR} and asks -- if it
is possible to construct for a quantum correlation scenario a simulation
model which would be (i) local in the sense meant by Einstein, Podolsky and
Rosen \cite{EPR}; (ii) similar by its measure-theoretic construction to the
concept of an LHV model and (iii) incorporate the latter only as a
particular case. This problem is also urgent for all multipartite
correlation scenarios (not necessarily quantum) specified not in terms of a
single probability space. The latter is one of the main notions of the
conventional probability theory.

Apart from the purely theoretical interest, such a local simulation model
could also single out a state parameter characterizing quantitatively
violations of Bell-type inequalities by a multipartite quantum state - the
problem discussed in the literature ever since the seminal result of
Tsirelson \cite{tsirelson}.

Note that though, for correlation bipartite Bell-type inequalities, quantum
violations are upper bounded \cite{tsirelson1} by the Grothendieck's
constant \cite{30} independently on a dimension of a bipartite quantum state
and numbers of settings and outcomes at each site, this is not already the
case for bipartite Bell-type inequalities on joint probabilities. Since
Bell-type inequalities are now widely used in many quantum information tasks 
\cite{Acin1, buhrman}, bounds on quantum violations of Bell-type
inequalities have been recently intensively discussed in the literature\
both computationally \cite{brunner} and theoretically \cite{junge, kaplan,
junge1, junge2, carlos} and it has been found \cite{junge} that some
tripartite quantum states "can lead to arbitrarily large" violations of
correlation Bell-type inequalities. For an $N$-partite quantum state, bounds
on violation of a Bell-type inequality of an arbitrary type (either on
correlation functions or on joint probabilities) have not been reported \cite%
{31} in the literature.

In the present paper, we introduce for the probabilistic description of a
general correlation scenario a new simulation model, \emph{a} \emph{local
quasi hidden variable (LqHV) model, }where locality and the
measure-theoretic structure inherent to an LHV model are preserved but
positivity of a simulation measure is dropped. We prove that every quantum
correlation scenario admits the simulation in LqHV terms and construct via
the LqHV approach analogs of Bell-type inequalities for an $N$-partite
quantum state. This allows us to find the new analytical and numerical upper
bounds on the maximal violation by an $N$-partite quantum state of all
Bell-type inequalities -- either on correlation functions or on joint
probabilities and for outcomes of an arbitrary spectral type, discrete or
continuous. The paper is organized as follows.

In section 2, for our consideration in sections 5, 6, we specify some new
dilation characteristics of an $N$-partite quantum state and discuss their
properties.

In section 3, we introduce for a general $N$-partite correlation scenario
with $S_{n}$ measurements at each $n$-th site the notion of an LqHV model
and specify a necessary and sufficient condition for LqHV modelling.

In section 4, we recall \cite{loubenets7} for an $S_{1}\times \cdots \times
S_{N}$-setting correlation scenario with outcomes of an arbitrary type,
discrete or continuous, the general form of all Bell-type inequalities --
either on correlation functions or on joint probabilities.

In section 5, we prove that every quantum $S_{1}\times \cdots \times S_{N}$%
-setting correlation scenario admits an LqHV model and introduce, for an $N$%
-partite quantum state, the exact analytical upper bound on the state
parameter specifying a possibility of its $S_{1}\times \cdots \times S_{N}$%
-setting LHV description.

In section 6, via the LqHV approach, we construct analogs of Bell-type
inequalities for an $N$-partite quantum state and find the new analytical
and numerical upper bounds on the maximal violation by an $N$-partite
quantum state of all $S_{1}\times \cdots \times S_{N}$-setting Bell-type
inequalities. The comparison of our exact general $N$-partite numerical
upper estimate specified for $N=2,3$ with the bipartite and tripartite
numerical estimates reported in the literature is given in section 6.2.

In section 7, we summarize the main results of the present paper.

In appendices A, B, C, we present proofs of some statements formulated in
sections 2, 5 and 6, respectively.

\section{Preliminaries: source operators, tensor positivity, the covering
norm}

In this section, for our consideration in sections 5, 6, we specify the
notion of a \emph{source operator} \cite{loubenets2, loubenets2a} for an $N$%
-partite state, the notion of \emph{tensor positivity} \cite{loubenets11}
and introduce a new norm, \emph{the covering norm}, on the space of all
self-adjoint trace class operators on a tensor product Hilbert space.

For a quantum state $\rho $ on a complex separable Hilbert space $\mathcal{H}%
_{1}\otimes \cdots \otimes \mathcal{H}_{N}$ and arbitrary positive integers $%
S_{1},...,S_{N}\geq 1,$ denote by $T_{_{S_{1}\times \cdots \times
S_{N}}}^{(\rho )}$ a self-adjoint trace class operator, defined on $\mathcal{%
H}_{1}^{\otimes S_{1}}\otimes \cdots \otimes \mathcal{H}_{N}^{\otimes S_{N}}$
and satisfying the relation%
\begin{align}
& \mathrm{tr}\left[ T_{_{S_{1}\times \cdots \times S_{N}}}^{(\rho )}\left\{ 
\mathbb{I}_{\mathcal{H}_{1}^{\otimes k_{1}}}\otimes X_{1}\otimes \mathbb{I}_{%
\mathcal{H}_{1}^{\otimes (S_{1}-1-k_{1})}}\otimes \cdots \otimes \mathbb{I}_{%
\mathcal{H}_{N}^{\otimes k_{N}}}\otimes X_{N}\otimes \mathbb{I}_{\mathcal{H}%
_{1}^{\otimes (S_{N}-1-k_{N})}}\right\} \right]  \label{d} \\
& =\mathrm{tr}\left[ \rho \left\{ X_{1}\otimes \cdots \otimes X_{N}\right\} %
\right] ,  \notag \\
k_{1}& =0,...,(S_{1}-1),...,k_{N}=0,...,(S_{N}-1),  \notag
\end{align}%
for all bounded linear operators $X_{1},...,X_{N}$ on Hilbert spaces $%
\mathcal{H}_{1},....,\mathcal{H}_{N},$ respectively. In (\ref{d}), we set $%
\mathbb{I}_{\mathcal{H}_{n}^{\otimes k}}\otimes X_{n}\mid _{_{k=0}}$ $%
=X_{n}\otimes \mathbb{I}_{\mathcal{H}_{n}^{\otimes k}}\mid _{_{k=0}}$ $%
:=X_{n}.$ Clearly, \textrm{tr}$[T_{_{S_{1}\times \cdots \times
S_{N}}}^{(\rho )}]=1$ and $T_{_{1\times \cdots \times 1}}^{(\rho )}\equiv
\rho .$

\begin{definition}[Source operators \protect\cite{loubenets2, loubenets2a}]
For a state $\rho $ on a Hilbert space \cite{32} $\mathcal{H}_{1}\otimes
\cdots \otimes \mathcal{H}_{N}$ and arbitrary positive integers $%
S_{1},...,S_{N}\geq 1$, we call each of self-adjoint trace class operators $%
T_{_{S_{1}\times \cdots \times S_{N}}}^{(\rho )}$ on $\mathcal{H}%
_{1}^{\otimes S_{1}}\otimes \cdots \otimes \mathcal{H}_{N}^{\otimes S_{N}}$
satisfying relation (\ref{d}) as an $S_{1}\times \cdots \times S_{N}$%
-setting source operator for state $\rho .$
\end{definition}

For a source operator $T,$ its the trace norm 
\begin{equation}
\left\Vert T\right\Vert _{1}:=\mathrm{tr}[|T|]=\mathrm{tr}\left[ T^{+}+T^{-}%
\right] =1+2\mathrm{tr}[T^{-}]\geq1.  \label{4}
\end{equation}
Here, $T^{\pm}\geq0,$ $T^{+}T^{-}=T^{-}T^{+}=0$ are positive trace class
operators in the spectral decomposition $T=T^{+}-T^{-}$ and $|T|:=\sqrt{T^{2}%
}=T^{+}+T^{-}.$

\begin{proposition}
For every state $\rho$ on a Hilbert space $\mathcal{H}_{1}\otimes\cdots
\otimes\mathcal{H}_{N}$ and arbitrary positive integers $S_{1},...S_{N}%
\geq1, $ there exists a source operator $T_{_{S_{1}\times\cdots\times
S_{N}}}^{(\rho )}.$
\end{proposition}

\begin{proof}
For a bipartite case, this statement has been proved for settings $1\times2$ 
$,$ $2\times1$ by proposition 1 in Ref. \cite{loubenets2}. This proof was
further generalized in appendix of Ref. \cite{loubenets3} for arbitrary $%
1\times S_{2},$ $S_{1}\times1.$ The proof for a general $N$-partite case
with setting $S_{1}\times...\times S_{N}$ is presented in appendix A.
\end{proof}

If $T_{_{S_{1}\times\cdots\times S_{N}}}^{(\rho)}$ is a source operator for
state $\rho$, then each of its reduced $\left( T_{_{S_{1}\times\cdots\times
S_{N}}}^{(\rho)}\right) _{red}$ on a Hilbert space $\mathcal{H}_{1}^{\otimes
L_{1}}\otimes\cdots\otimes\mathcal{H}_{N}^{\otimes L_{N}}$, with $1\leq
L_{n}<S_{n},$ constitutes an $L_{1}\times\cdots\times L_{N}$-setting source
operator for state $\rho$ and 
\begin{equation}
1\leq\left\Vert \left( T_{_{S_{1}\times\cdots\times S_{N}}}^{(\rho)}\right)
_{red}\right\Vert _{1}\leq\left\Vert T_{_{S_{1}\times\cdots\times
S_{N}}}^{(\rho)}\right\Vert _{1}.  \label{6_}
\end{equation}
\medskip

In order to analyze situations where, for a source operator $T,$ relation $%
\mathrm{tr}[T\{X_{1}\otimes\cdots$ $\otimes X_{m}\}]\geq0$ holds for
arbitrary positive operators $X_{1},...,X_{m},$ we specify the following
general notion.

\begin{definition}[Tensor positivity \protect\cite{loubenets11}]
We call a bounded linear operator $Z$ on a Hilbert space $\mathcal{G}%
_{1}\otimes\mathcal{\cdots}\otimes\mathcal{G}_{m},$ $m\geq1,$ as tensor
positive and denote it by $Z\overset{\otimes}{\geq }0$ if the scalar product%
\begin{equation}
\left( \psi_{1}\otimes\cdots\otimes\psi_{m},Z\text{ }\psi_{1}\otimes
\cdots\otimes\psi_{m}\right) \geq0  \label{7}
\end{equation}
for arbitrary $\psi_{1}\in\mathcal{G}_{1},...,\psi_{m}\in\mathcal{G}_{m}.$
\end{definition}

\begin{remark}
For space $\mathcal{G}_{1}\otimes\mathcal{G}_{2},$ the notion of tensor
positivity is similar by its meaning to "block-positivity" in Ref. \cite%
{block}. We, however, consider that, for a tensor product of any number of
arbitrary Hilbert spaces, possibly infinite dimensional, our term "tensor
positivity" is more suitable.
\end{remark}

For $m=1$, tensor positivity is equivalent to positivity. For $m\geq2,$
positivity implies tensor positivity but not vice versa. For example,
operator $V(\psi_{1}\otimes\psi_{2}):=\psi_{2}\otimes\psi_{1}$ on space $%
\mathcal{H}\otimes\mathcal{H}$ is tensor positive but not positive.

Since on a complex separable Hilbert space, every positive operator is
self-adjoint, from (\ref{7}) and the spectral theorem it follows that, for a
trace class tensor positive operator $W\overset{\otimes}{\geq}0$ on a
complex separable Hilbert space $\mathcal{G}_{1}\otimes\mathcal{\cdots}%
\otimes \mathcal{G}_{m},$ relation $\mathrm{tr}\left[ W\{X_{1}\otimes\cdots%
\otimes X_{m}\}\right] \geq0$ holds for arbitrary positive operators $%
X_{1},...,X_{m}$ on spaces $\mathcal{G}_{1},..,\mathcal{G}_{m}$,
respectively. In particular, $\mathrm{tr}[W]\geq0,$ for each trace class $W%
\overset{\otimes }{\geq}0.$

If a trace class operator on $\mathcal{G}_{1}\otimes\mathcal{\cdots}\otimes%
\mathcal{G}_{m}$ is tensor positive, then any of its reduced operators is
also tensor positive. The converse of this statement is not true.

Coming back to source operators, we stress that though, for every $N$%
-partite state, a source operator exists (see proposition 1) for every
setting $S_{1}\times\cdots\times S_{N},$ \emph{an arbitrary }$N$\emph{%
-partite} \emph{state does not need to have a tensor positive source operator%
\cite{ter}. }

For example, every separable $N$-partite state $\rho=\sum\alpha_{i}\rho
_{1}^{(i)}\otimes\cdots\otimes\rho_{N}^{(i)},$ $\alpha_{i}>0,$ $\sum\alpha
_{i}=1,$ has a positive source operator 
\begin{equation}
\sum\alpha_{i}\left( \rho_{1}^{(i)}\right) ^{\otimes S_{1}}\otimes
\cdots\otimes\left( \rho_{N}^{(i)}\right) ^{\otimes S_{N}}  \label{9}
\end{equation}
for arbitrary $S_{1},...,S_{N}\geq1.$ However, a nonseparable state does not
need to have a tensor positive source operator even for at least one
setting. In Refs. \cite{loubenets2, loubenets2a, loubenets2b, loubenets3,
loubenets11}, we present examples of source operators for some nonseparable
bipartite states and single out the state parameters for which these source
operators become tensor positive.

Suppose now that we want to decompose a source operator into two tensor
positive operators. The following \emph{new} notion allows us to consider
such decompositions.

\begin{definition}[Coverings]
For a self-adjoint bounded linear operator $Z$ on a Hilbert space $\mathcal{G%
}_{1}\otimes \mathcal{\cdots }\otimes \mathcal{G}_{m},$ $m\geq 1,$ we call a
tensor positive operator $Z_{cov}$ on $\mathcal{G}_{1}\otimes \mathcal{%
\cdots }\otimes \mathcal{G}_{m}$ satisfying relations 
\begin{equation}
Z_{cov}\pm Z\overset{\otimes }{\geq }0  \label{10}
\end{equation}%
as a covering of $Z.$
\end{definition}

If $Z$ is tensor positive, then it, itself, represents one of its coverings.
In view of (\ref{10}), every self-adjoint bounded linear operator $Z$ on $%
\mathcal{G}_{1}\otimes \mathcal{\cdots }\otimes \mathcal{G}_{m}$ admits the
decomposition%
\begin{equation}
Z=\frac{1}{2}\left( Z_{cov}+Z\right) -\frac{1}{2}\left( Z_{cov}-Z\right)
\label{11}
\end{equation}%
via tensor positive operators $Z_{cov}\pm Z$ $\overset{\otimes }{\geq }0,$
where $Z_{cov}$ is any of its coverings.

For a source operator, we are interested in its trace class coverings.
Denote by $\mathcal{T}_{\mathcal{G}_{1}\otimes\mathcal{\cdots}\otimes%
\mathcal{G}_{m}}$ the linear space of all trace class operators on a Hilbert
space $\mathcal{G}_{1}\otimes\mathcal{\cdots}\otimes\mathcal{G}_{m}$ and by $%
\mathcal{T}_{\mathcal{G}_{1}\otimes\mathcal{\cdots}\otimes\mathcal{G}%
_{m}}^{(sa)}\subset\mathcal{T}_{\mathcal{G}_{1}\otimes\mathcal{\cdots}\otimes%
\mathcal{G}_{m}}$ -- the subspace of all self-adjoint trace class operators.

\begin{proposition}
For every self-adjoint trace class operator $W$ on a Hilbert space $\mathcal{%
G}_{1}\otimes\mathcal{\cdots}\otimes\mathcal{G}_{m},$ there exists a trace
class covering $W_{cov}.$
\end{proposition}

\begin{proof}
For an operator $W\in\mathcal{T}_{\mathcal{G}_{1}\otimes\mathcal{\cdots }%
\otimes\mathcal{G}_{m}}^{sa},$ consider its spectral decomposition $%
W=W^{+}-W^{-},$ where $W^{\pm}\geq0,$ $W^{+}W^{-}=W^{-}W^{+}=0.$ The
positive trace class operator $\left\vert W\right\vert :=\sqrt{W^{2}}%
=W^{+}+W^{-}$, with $\mathrm{tr}[\left\vert W\right\vert ]:=\left\Vert
W\right\Vert _{1}<\infty,$ constitutes a trace class covering of $W$. This
proves the statement.
\end{proof}

From (\ref{10}) it follows that, for every self-adjoint trace class operator 
$W,$ the relation 
\begin{equation}
\mathrm{tr}[W_{cov}]\geq \left\vert \mathrm{tr}[W]\right\vert \geq 0
\label{12---}
\end{equation}%
holds for each of its trace class coverings $W_{cov}.$

Thus, for every $W\in\mathcal{T}_{\mathcal{G}_{1}\otimes\mathcal{\cdots }%
\otimes\mathcal{G}_{m}}^{(sa)},$ the set $\{W_{cov}\in\mathcal{T}_{\mathcal{G%
}_{1}\otimes\mathcal{\cdots}\otimes\mathcal{G}_{m}}\}\ni\left\vert
W\right\vert $ contains at least one element and $\mathrm{tr}[W_{cov}]\geq0$
for each covering $W_{cov}$. Therefore, we can introduce on space $\mathcal{T%
}_{\mathcal{G}_{1}\otimes\mathcal{\cdots}\otimes\mathcal{G}_{m}}^{(sa)}$ the
following function 
\begin{equation}
f(W):=\inf_{W_{cov}\in\mathcal{T}_{\mathcal{G}_{1}\otimes\mathcal{\cdots }%
\otimes\mathcal{G}_{m}}}\mathrm{tr}\left[ W_{cov}\right] \geq0,\text{ \ \ \ }%
\forall W\in\mathcal{T}_{\mathcal{G}_{1}\otimes\mathcal{\cdots}\otimes%
\mathcal{G}_{m}}^{(sa)},  \label{13}
\end{equation}
and relations 
\begin{align}
f(W) & =0\text{ \ }\Leftrightarrow\text{ \ }W=0,  \label{14} \\
\text{\ }f(\alpha W) & =\left\vert \alpha\right\vert \text{ }f(W),\text{ \ \ 
}\forall\alpha\in\mathbb{R},  \notag \\
f(W_{1}+W_{2}) & \leq f(W_{1})+f(W_{2}),  \notag
\end{align}
hold for all $W,W_{1},W_{2}\in\mathcal{T}_{\mathcal{G}_{1}\otimes \mathcal{%
\cdots}\otimes\mathcal{G}_{m}}^{(sa)}$. The first of these relations follows
from property 1 in lemma 1 below. The second relation -- from (\ref{10}), (%
\ref{13}). For the proof of the last relation in (\ref{14}), we note that,
for arbitrary trace class coverings $(W_{1})_{cov},$ $(W_{2})_{cov}$ of
operators $W_{1},$ $W_{2}\in\mathcal{T}_{\mathcal{G}_{1}\otimes\mathcal{%
\cdots}\otimes\mathcal{G}_{m}}^{sa},$ operator $(W_{1})_{cov}+(W_{2})_{cov} $
constitutes a possible trace class covering of $W_{1}+W_{2}$. Hence, $%
\{(W_{1}+W_{2})_{cov}\}\supseteq\{(W_{1})_{cov}+(W_{2})_{cov}\}$ and, taking
this inclusion into the account in infimum (\ref{13}) specifying $%
f(W_{1}+W_{2}),$ we come to the third relation in (\ref{14}).

In view of relations (\ref{14}), function (\ref{13}) constitutes a norm on
space $\mathcal{T}_{\mathcal{G}_{1}\otimes \mathcal{\cdots }\otimes \mathcal{%
G}_{m}}^{(sa)}$

\begin{definition}
We refer to norm (\ref{13}) as the covering norm and denote it by 
\begin{equation}
\left\Vert W\right\Vert _{cov}:=\inf_{W_{cov}\in \mathcal{T}_{\mathcal{G}%
_{1}\otimes \mathcal{\cdots }\otimes \mathcal{G}_{m}}}\mathrm{tr}[W_{cov}],%
\text{ \ \ }\forall W\in \mathcal{T}_{\mathcal{G}_{1}\otimes \mathcal{\cdots 
}\otimes \mathcal{G}_{m}}^{(sa)}.  \label{15}
\end{equation}
\end{definition}

The following general properties of the covering norm are proved in Appendix
A.

\begin{lemma}
For a self-adjoint trace class operator $W$ on a Hilbert space $\mathcal{G}%
_{1}\otimes\mathcal{\cdots}\otimes\mathcal{G}_{m}:\medskip$\newline
(1) $\left\vert \mathrm{tr}[W]\right\vert \leq\sup\left\vert \mathrm{tr}%
\left[ W\{X_{1}\otimes\cdots\otimes X_{m}\}\right] \right\vert
\leq\left\Vert W\right\Vert _{cov}\leq\left\Vert W\right\Vert _{1},$ \newline
where supremum is taken over all self-adjoint bounded linear operators $%
X_{1},...,X_{m}$ with operator norms $||X_{j}||=1$ on spaces $\mathcal{G}%
_{1},...,\mathcal{G}_{m},$ respectively;$\medskip$\newline
(2) if $W\overset{\otimes}{\geq}0,$ then $\left\Vert W\right\Vert _{cov}=%
\mathrm{tr}[W];\medskip$\newline
(3) $\left\vert \mathrm{tr}[W]\right\vert \leq\left\Vert W_{red}\right\Vert
_{cov}\leq\left\Vert W\right\Vert _{cov}$ for each operator $W_{red}$
reduced from $W.$
\end{lemma}

As an example, consider the self-adjoint operator $V(\psi_{1}\otimes\psi
_{2})=\psi_{2}\otimes\psi_{1}$ on $\mathbb{C}^{d}\otimes\mathbb{C}^{d}.$ For
this operator, the covering norm $\left\Vert V\right\Vert _{cov}=d$ while
the trace norm $\left\Vert V\right\Vert _{1}=d^{2}.$

For an $S_{1}\times ...\times S_{N}$-setting source operator $%
T_{_{S_{1}\times ...\times S_{N}}}^{(\rho )}$ for a state $\rho $ on a
Hilbert space $\mathcal{H}_{1}\otimes \cdots \otimes \mathcal{H}_{N},$ lemma
1 implies%
\begin{align}
1& \leq \left\Vert T_{_{S_{1}\times \cdots \times S_{N}}}^{(\rho
)}\right\Vert _{cov}\leq \left\Vert T_{_{S_{1}\times \cdots \times
S_{N}}}^{(\rho )}\right\Vert _{1},  \label{19} \\
T_{_{S_{1}\times \cdots \times S_{N}}}^{(\rho )}\overset{\otimes }{\geq }0%
\text{ \ }& \Rightarrow \text{ \ }\left\Vert T_{_{S_{1}\times \cdots \times
S_{N}}}^{(\rho )}\right\Vert _{cov}=1  \label{19'}
\end{align}%
and 
\begin{equation}
1\leq \left\Vert \left( T_{_{S_{1}\times \cdots \times S_{N}}}^{(\rho
)}\right) _{_{red}}\right\Vert _{cov}\leq \left\Vert T_{_{S_{1}\times \cdots
\times S_{N}}}^{(\rho )}\right\Vert _{cov}  \label{20}
\end{equation}%
for each source operator $\left( T_{_{S_{1}\times \cdots \times
S_{N}}}^{(\rho )}\right) _{_{red}}$ reduced from a source operator $%
T_{_{S_{1}\times \cdots \times S_{N}}}^{(\rho )}.$

\section{LqHV modelling of a general correlation scenario}

Consider an $N$-partite correlation scenario, where each $n$-th of $N\geq2$
parties (players) performs $S_{n}\geq1$ measurements with outcomes $%
\lambda_{n}\in\Lambda_{n}\ $of an arbitrary type \ and $\mathcal{F}%
_{\Lambda_{n}} $ is a $\sigma$-algebra of events $F_{n}\subseteq$ $\Lambda
_{n}$ observed at $n$-th site. For the general framework on the
probabilistic description of multipartite correlation scenarios, see Ref. 
\cite{loubenets3}

We label each measurement at $n$-th site by a positive integer $%
s_{n}=1,...,S_{n}$ and each of $N$-partite joint measurements, induced by
this correlation scenario and with outcomes $(\lambda_{1},\ldots,%
\lambda_{N})\in\Lambda_{1}\times\cdots\times\Lambda_{N}$ -- by an $N$-tuple $%
(s_{1},...,s_{N}),$ where $n$-th component refers to a measurement at $n$-th
site.

For concreteness, we further refer to an $S_{1}\times \cdots \times S_{N}$%
-setting correlation scenario with outcomes in $\Lambda _{1}\times \cdots
\times \Lambda _{N}$ by symbol 
\begin{equation}
\mathcal{E}_{S,\Lambda },\text{ \ \ }S:=S_{1}\times \cdots \times S_{N},%
\text{ \ \ }\Lambda :=\Lambda _{1}\times \cdots \times \Lambda _{N},
\label{notation}
\end{equation}%
and denote by $P_{(s_{1},...,s_{N})}^{(\mathcal{E}_{S,\Lambda })}$ a
probability measure, describing an $N$-partite joint measurement $%
(s_{1},...,s_{N})$ of scenario $\mathcal{E}_{S,\Lambda }$ and defined on the
direct product $(\Lambda _{1}\times \cdots \times \Lambda _{N},$ $\mathcal{F}%
_{\Lambda _{1}}\otimes \cdots \otimes \mathcal{F}_{\Lambda _{N}})$ of
measurable spaces $(\Lambda _{n},\mathcal{F}_{\Lambda _{n}}),$ $n=1,...,N.$
Recall \cite{dunford} that the product $\sigma $-algebra $\mathcal{F}%
_{\Lambda _{1}}\otimes \cdots \otimes \mathcal{F}_{\Lambda _{N}}$ is the
smallest $\sigma $-algebra generated by the set of all rectangles $%
F_{1}\times \cdots \times F_{N}\subseteq \Lambda _{1}\times \cdots \times
\Lambda _{N}$ with measurable "sides" $F_{n}\in \mathcal{F}_{\Lambda _{n}},$ 
$n=1,...,N.$

In what follows, we consider only standard measurable spaces. In this case,
each $(\Lambda _{n},\mathcal{F}_{\Lambda _{n}})$ is Borel isomorphic to a
measurable space $(\mathcal{X}_{n},\mathcal{B}_{\mathcal{X}_{n}}),$ where $%
\mathcal{X}_{n}\mathcal{\in B}_{\mathbb{R}}$ is a Borel subset of $\mathbb{R}
$ and $\mathcal{B}_{\mathcal{X}_{n}}:=\mathcal{B}_{\mathbb{R}}\cap \mathcal{X%
}_{n}$ is the trace on $\mathcal{X}_{n}$ of the Borel $\sigma $-algebra $%
\mathcal{B}_{\mathbb{R}}$ on $\mathbb{R}.$

For a general correlation scenario $\mathcal{E}_{S,\Lambda},$ let us
introduce the following new type of simulation models.

\begin{definition}
We say that an $S_{1}\times ...\times S_{N}$-setting correlation scenario $%
\mathcal{E}_{S,\Lambda }$, with joint probability distributions $%
P_{(s_{1},...,s_{N})}^{(\mathcal{E}_{S,\Lambda })},$ \ $%
s_{1}=1,...,S_{1},...,s_{N}=1,...,S_{N},$ and outcomes $(\lambda _{1},\ldots
,\lambda _{N})\in \Lambda _{1}\times \cdots \times \Lambda _{N}:=\Lambda ,$
admits a local quasi hidden variable (LqHV) model if all of its joint
probability distributions admit the representation 
\begin{align}
P_{(s_{1},...,s_{N})}^{(\mathcal{E}_{S,\Lambda })}(F_{1}\times \cdots \times
F_{N})& =\dint\limits_{\Omega }P_{1}^{(s_{1})}(F_{1}|\omega )\cdot \ldots
\cdot P_{N}^{(s_{_{N}})}(F_{N}|\omega )\text{ }\nu _{\mathcal{E}_{S,\Lambda
}}(\mathrm{d}\omega ),  \label{21} \\
F_{1}& \in \mathcal{F}_{\Lambda _{1}},...,F_{N}\in \mathcal{F}_{\Lambda
_{N}},  \notag
\end{align}%
in terms of a single measure space $\left( \Omega ,\mathcal{F}_{\Omega },\nu
_{\mathcal{E}_{S,\Lambda }}\right) ,$ with a normalized bounded real-valued
measure $\nu _{\mathcal{E}_{S,\Lambda }},$ and conditional probability
measures $P_{n}^{(s_{_{n}})}(\cdot |\omega ):\mathcal{F}_{\Lambda
_{n}}\rightarrow \lbrack 0,1],$ defined $\nu _{_{\mathcal{E}_{S,\Lambda }}}$%
-a.e. (almost everywhere) on $\Omega $ and such that, for each $%
s_{n}=1,...,S_{n}$ and every $n=1,...,N,$ function $%
P_{n}^{(s_{_{n}})}(F_{n}|\cdot ):\Omega \rightarrow \lbrack 0,1]$ is
measurable for all $F_{n}\in \mathcal{F}_{\Lambda _{n}}.$
\end{definition}

\begin{notation}
In a triple $\left( \Omega,\mathcal{F}_{\Omega},\nu\right) $ representing a
measure space, $\Omega$ is a non-empty set, $\mathcal{F}_{\Omega}$ is a $%
\sigma$-algebra of subsets of $\Omega$ and $\nu$ is a measure on a
measurable space $\left( \Omega,\mathcal{F}_{\Omega}\right) .$ A real-valued
measure $\nu$ is called normalized if $\nu(\Omega)=1$ and bounded if $%
\left\vert \nu(F)\right\vert \leq M<\infty$ for all $F\in\mathcal{F}%
_{\Omega}.$
\end{notation}

We stress that, in an LqHV model (\ref{21}), measure $\nu _{\mathcal{E}%
_{S,\Lambda }}$ has a simulation character and may, in general, depend (via
the lower index $\mathcal{E}_{S,\Lambda }$) on measurement settings at all
(or some) sites, as an example, see measure (\ref{22}).

From (\ref{21}) it follows that a correlation scenario $\mathcal{E}%
_{S,\Lambda }$ admitting an LqHV model satisfies the general nonsignaling
condition specified by definition 1 (Eq. (10)) in Ref. \cite{loubenets3}.

If, for a correlation scenario $\mathcal{E}_{S,\Lambda }$, there exists
representation (\ref{21}), where a normalized real-valued measure $\nu _{%
\mathcal{E}_{S,\Lambda }}$ is positive and, hence, is a probability measure,
then this scenario admits an LHV\ model\emph{\ }formulated\emph{\ }for a
general case by definition 4 (Eq. (26)) in Ref. \cite{loubenets3}.

\begin{remark}
Recall \cite{dunford} that a bounded real-valued measure $\nu $ on a
measurable space $(\Omega ,\mathcal{F}_{\Omega })$ admits the Jordan
decomposition $\nu =\nu ^{+}-\nu ^{=}$ via positive measures%
\begin{equation}
\nu ^{+}(F):=\sup_{F^{\prime }\in \mathcal{F}_{\Omega },F^{\prime }\subseteq
F}\nu (F^{\prime }),\text{ \ \ \ }\nu ^{-}(F):=-\inf_{F^{\prime }\in 
\mathcal{F}_{\Omega },F^{\prime }\subseteq F}\nu (F^{\prime }),\text{ \ \ }%
\forall F\in \mathcal{F}_{\Omega },  \label{33}
\end{equation}%
with disjoint supports. The sum $(\nu ^{+}(\Omega )+\nu ^{-}(\Omega ))$
coincides with the total variation $\left\vert \nu \right\vert (\Omega )$ of
measure $\nu $ on $\Omega ,$ which is defined by relation 
\begin{equation}
\sup \sum_{i=1}^{m}\left\vert \nu (F_{i})\right\vert :=\left\vert \nu
\right\vert (\Omega )\equiv \left\Vert \nu \right\Vert _{var},  \label{tv}
\end{equation}%
where supremum is taken over all finite systems $\{F_{i}\}$ of disjoint sets
in $\mathcal{F}_{\Omega }.$ For a bounded measure $\nu ,$ its total
variation $\left\Vert \nu \right\Vert _{var}<\infty $ and $\left\Vert \cdot
\right\Vert _{var}$ constitutes\textrm{\ }a norm, the total variation norm,%
\textrm{\ }on the linear space of all bounded real-valued measures on a
measurable space $(\Omega ,\mathcal{F}_{\Omega }).$ Thus, for a bounded
real-valued measure $\nu ,$ we have 
\begin{equation}
\nu ^{+}(\Omega )+\nu ^{-}(\Omega )=\left\Vert \nu \right\Vert _{var}.
\label{35''}
\end{equation}
If a bounded real-valued measure $\nu $ is normalized, then 
\begin{equation}
\left\Vert \nu \right\Vert _{var}=1+2\nu ^{-}(\Omega )\geq 1.  \label{35}
\end{equation}%
A normalized bounded real-valued measure $\nu $\ is a probability measure
iff $\left\Vert \nu \right\Vert _{var}=1.$ Note that relation%
\begin{equation}
\sup_{F\in \mathcal{F}_{\Omega }}\left\vert \nu (F)\right\vert \leq
\left\Vert \nu \right\Vert _{var}\leq 2\sup_{F\in \mathcal{F}_{\Omega
}}\left\vert \nu (F)\right\vert  \label{var}
\end{equation}%
holds for every real-valued measure $\nu .$
\end{remark}

From the Jordan decomposition for measure $\nu _{\mathcal{E}_{S,\Lambda }}$
it follows that if a correlation scenario admits an LqHV model (\ref{21}),
then each of its joint probability distributions $P_{(s_{1},...,s_{N})}^{(%
\mathcal{E}_{S,\Lambda })}$ can be expressed via the affine combination of
some LHV distributions $P_{(s_{1},...,s_{N})}^{(\mathcal{E}_{S,\Lambda
}^{lhv})}$ that are represented in (\ref{21}) by the same conditional
measures $P_{n}^{(s_{_{n}})}(\cdot |\omega )$. On the other hand, if a
correlation scenario with a finite number of outcomes at each site admits
the affine model in the sense of Ref. \cite{kaplan}, then this scenario
admits the special LqHV model, where measure $\nu _{\mathcal{E}_{S,\Lambda
}} $ is given by the affine combination of discrete probability measures and
each $P_{n}^{(s_{_{n}})}(F_{n}|\omega )$ has the particular form $\chi
_{_{f_{n,s_{n}}^{-1}(F_{n})}}(\omega ),$ $F_{n}\in \mathcal{F}_{\Lambda
_{n}},$ where $f_{n,s_{n}}:\Omega \rightarrow \Lambda _{n}$ is some
measurable function, $f_{n,s_{n}}^{-1}(F_{n}):=\{\omega \in \Omega \mid $ $%
f_{n,s_{n}}(\omega )\in F_{n}\}$ and $\chi _{D}(\cdot )$ is the indicator
function of a subset $D\subseteq \Omega ,$ that is: $\chi _{D}(\omega )=1$
if $\omega \in D$ and $\chi _{D}(\omega )$ $=0$ if $\omega \notin D.$

Thus, an LqHV model incorporates as particular cases and generalizes in one
whole both types of simulation models discussed in the literature -- an LHV
model and an affine model. Note that the latter model is, in principle,
built up on the concept of an LHV model.

We stress that, \emph{in an LqHV model, locality and the measure-theoretic
structure inherent to an LHV model are preserved}.

The following general theorem introduces a necessary and sufficient
condition for LqHV modelling.

\begin{theorem}
An $S_{1}\times ...\times S_{N}$-setting correlation scenario $\mathcal{E}%
_{S,\Lambda }$ admits an LqHV model (\ref{21}) if and only if, on the direct
product space $(\Lambda _{1}^{S_{1}}\times \cdots \times \Lambda
_{N}^{S_{N}},\mathcal{F}_{\Lambda _{1}}^{\otimes S_{1}}\otimes \cdots
\otimes \mathcal{F}_{\Lambda _{N}}^{\otimes S_{N}}),$ there exists a
normalized bounded real-valued measure \cite{33}%
\begin{align}
& \mu _{\mathcal{E}_{S,\Lambda }}\left( \mathrm{d}\lambda _{1}^{(1)}\times
\cdots \times \mathrm{d}\lambda _{1}^{(S_{1})}\times \cdots \times \mathrm{d}%
\lambda _{N}^{(1)}\times \cdots \times \mathrm{d}\lambda
_{N}^{(S_{N})}\right) ,  \label{26} \\
\lambda _{n}^{(s_{n})}& \in \Lambda _{n},\text{ \ \ }s_{n}=1,...,S_{n},\text{
\ \ }n=1,...,N,  \notag
\end{align}%
returning all joint probability distributions $P_{(s_{1},...,s_{N})}^{(%
\mathcal{E}_{S,\Lambda })}$ of scenario $\mathcal{E}_{S,\Lambda }$ as the
corresponding marginals.
\end{theorem}

\begin{proof}
Let scenario $\mathcal{E}_{S,\Lambda }$ admit an LqHV model (\ref{21}) \
Then the normalized real-valued measure 
\begin{equation}
\dint\limits_{\Omega }\text{ }\{\dprod\limits_{s_{n}=1,...,S_{n},\text{ }%
n=1,...,N}P_{n}^{(s_{n})}(\mathrm{d}\lambda _{n}^{(s_{n})}\mid \omega )\text{
}\}\text{ }\nu _{\mathcal{E}_{S,\Lambda }}(\mathrm{d}\omega )  \label{meas}
\end{equation}%
on $(\Lambda _{1}^{S_{1}}\times \cdots \times \Lambda _{N}^{S_{N}},\mathcal{F%
}_{\Lambda _{1}}^{\otimes S_{1}}\otimes \cdots \otimes \mathcal{F}_{\Lambda
_{N}}^{\otimes S_{N}})$ returns all distributions $P_{(s_{1},...,s_{N})}^{(%
\mathcal{E}_{S,\Lambda })}$ of scenario $\mathcal{E}_{S,\Lambda }$ as the
corresponding marginals. The total variation of measure (\ref{meas}) is
upper bounded by $\left\Vert \nu _{\mathcal{E}_{S,\Lambda }}\right\Vert
_{var}<\infty ,$ so that, in view of relation (\ref{var}), this measure is
bounded.

In order to prove the sufficiency part of theorem 1, let there exist a
normalized bounded real-valued measure $\widetilde{\mu }_{\mathcal{E}%
_{S,\Lambda }}$ returning all probability distributions $%
P_{(s_{1},...,s_{N})}^{(\mathcal{E}_{S,\Lambda })}$ of scenario $\mathcal{E}%
_{S,\Lambda }$ as the corresponding marginals. This means that the
representation 
\begin{align}
P_{(s_{1},...,s_{N})}^{(\mathcal{E}_{S,\Lambda })}\left( F_{1}\times \cdots
\times F_{N}\right) & =\dint \chi _{_{F_{1}}}(\lambda _{1}^{(s_{1})})\cdot
\ldots \cdot \chi _{_{F_{N}}}(\lambda _{N}^{(s_{N})})\text{ }\widetilde{\mu }%
_{\mathcal{E}_{S,\Lambda }}(\mathrm{d}\lambda _{1}^{(1)}  \label{25} \\
& \times \cdots \times \mathrm{d}\lambda _{1}^{(S_{1})}\times \cdots \times 
\mathrm{d}\lambda _{N}^{(1)}\times \cdots \times \mathrm{d}\lambda
_{N}^{(S_{N})}),  \notag \\
s_{1}& =1,...,S_{1},...,s_{N}=1,...,S_{N},  \notag
\end{align}%
holds for all $F_{1}\in \mathcal{F}_{\Lambda _{1}},...,F_{N}\in \mathcal{F}%
_{\Lambda _{N}}.$ Representation (\ref{25}) constitutes a particular case of
the LqHV representation (\ref{21}) specified with 
\begin{align}
\omega ^{\prime }& =\left( \lambda _{1}^{(1)},...,\lambda
_{1}^{(S_{1})},...,\lambda _{N}^{(1)},...,\lambda _{N}^{(S_{N})}\right) , \\
\Omega ^{\prime }& =\Lambda _{1}^{S_{1}}\times ...\times \Lambda
_{N}^{S_{N}},\text{ \ \ }\mathcal{F}_{\Omega ^{\prime }}=\mathcal{F}%
_{\Lambda _{1}}^{\mathcal{\otimes S}_{1}}\otimes \cdots \otimes \mathcal{F}%
_{\Lambda _{N}}^{\otimes S_{N}},  \notag \\
\nu _{\mathcal{E}_{S,\Lambda }}^{\prime }& =\widetilde{\mu }_{\mathcal{E}%
_{S,\Lambda }},\text{ \ \ \ }P_{n}^{(s_{n})}(F_{n}|\omega ^{\prime })=\chi
_{F_{n}}(\lambda _{n}^{(s_{_{n}})}).  \notag
\end{align}%
This proves the statement.
\end{proof}

The following corollary of theorem 1 corresponds to the statements (a), (c)
of the general theorem 1 on LHV modelling in Ref. \cite{loubenets3}.

\begin{corollary}
An $S_{1}\times ...\times S_{N}$-setting correlation scenario $\mathcal{E}%
_{S,\Lambda }$ admits an LHV model if and only if there exists a probability
measure $\mu _{\mathcal{E}_{S,\Lambda }}^{\prime }$ on space $(\Lambda
_{1}^{S_{1}}\times \cdots \times \Lambda _{N}^{S_{N}},\mathcal{F}_{\Lambda
_{1}}^{\otimes S_{1}}\otimes \cdots \otimes \mathcal{F}_{\Lambda
_{N}}^{\otimes S_{N}})$ returning all joint probability distributions $%
P_{(s_{1},...,s_{N})}^{(\mathcal{E}_{S,\Lambda })}$ of scenario $\mathcal{E}%
_{S,\Lambda }$ as the corresponding marginals.
\end{corollary}

\begin{proof}
If scenario $\mathcal{E}_{S,\Lambda }$ admits an LHV model, then there
exists representation (\ref{21}) with some probability measure $\nu ^{\prime
}$ and, for this $\nu ^{\prime },$ the constructed normalized measure (\ref%
{meas}) is a probability one. Conversely, let there exist a probability
measure $\mu _{\mathcal{E}_{S,\Lambda }}^{\prime }$ returning all
distributions $P_{(s_{1},...,s_{N})}^{(\mathcal{E}_{S,\Lambda })}$ of
scenario $\mathcal{E}_{S,\Lambda }$ as the corresponding marginals. Then
representation (\ref{25}) with probability measure $\mu _{\mathcal{E}%
_{S,\Lambda }}^{\prime }$ constitutes a particular LHV model.
\end{proof}

\section{Bell-type inequalities}

For an $S_{1}\times\cdots\times S_{N}$-setting correlation scenario $%
\mathcal{E}_{S,\Lambda},$ with joint probability distributions $%
P_{(s_{1},...,s_{N})}^{(\mathcal{E}_{S,\Lambda})}$ and outcomes $(\lambda
_{1},...,\lambda_{N})\in\Lambda_{1}\times\cdots\times\Lambda_{N}:=\Lambda,$
consider a linear combination%
\begin{equation}
\dsum \limits_{s_{1},...,s_{_{N}}}\left\langle
\psi_{(s_{1},\ldots,s_{_{N}})}(\lambda_{1},...,\lambda _{N})\right\rangle _{%
\mathcal{E}_{S,\Lambda}}  \label{38}
\end{equation}
of averages 
\begin{align}
& \left\langle \psi_{(s_{1},\ldots,s_{_{N}})}(\lambda_{1},...,\lambda
_{N})\right\rangle _{\mathcal{E}_{S,\Lambda}}  \label{av} \\
& :=\dint
\limits_{\Lambda}\psi_{(s_{1},\ldots,s_{_{N}})}(\lambda_{1},...,%
\lambda_{N})P_{(s_{1},...,s_{N})}^{(\mathcal{E}_{S,\Lambda})}(\mathrm{d}%
\lambda_{1}\times \cdots\times\mathrm{d}\lambda_{N})  \notag
\end{align}
arising under joint measurements $(s_{1},\ldots,s_{_{N}})$ and specified by
a family 
\begin{equation}
\Psi_{S,\Lambda}:=\left\{ \psi_{(s_{1},\ldots,s_{_{N}})},\text{ \ }%
s_{1}=1,...,S_{1},...,s_{N}=1,...,S_{N}\right\}
\end{equation}
of bounded measurable real-valued functions $\psi_{(s_{1},\ldots,s_{_{N}})}(%
\lambda_{1},...,\lambda_{N})$ on the direct product $(\Lambda_{1}\times%
\cdots\times\Lambda_{N},\mathcal{F}_{\Lambda_{1}}\otimes\cdots \otimes%
\mathcal{F}_{\Lambda_{N}})$ of measurable spaces $(\Lambda _{n},\mathcal{F}%
_{\Lambda_{n}}),$ $n=1,...,N.$

If, in (\ref{av}), function $\psi _{(s_{1},\ldots ,s_{_{N}})}$ has the
product form $\phi _{s_{1}}(\lambda _{1})\cdot ...\cdot \phi
_{s_{_{N}}}(\lambda _{N}),$ then, depending on a concrete choice of
functions $\phi _{s_{n}}(\lambda _{n}),$ for a joint measurement $%
(s_{1},...,s_{_{N}}),$ the average 
\begin{equation}
\dint\limits_{\Lambda }\phi _{s_{1}}(\lambda _{1})\cdot ...\cdot \phi
_{s_{_{N}}}(\lambda _{N})\text{ }P_{(s_{1},...,s_{N})}^{(\mathcal{E}%
_{S,\Lambda })}(\mathrm{d}\lambda _{1}\times \cdots \times \mathrm{d}\lambda
_{N})  \label{41}
\end{equation}%
may refer to either the joint probability 
\begin{equation}
\left\langle \chi _{F_{1}}(\lambda _{1}^{(s_{1})})\cdot ...\cdot \chi
_{F_{N}}(\lambda _{N}^{(s_{N})})\right\rangle _{\mathcal{E}_{S,\Lambda
}}=P_{(s_{1},...,s_{N})}^{(\mathcal{E}_{S,\Lambda })}(F_{1}\times \cdots
\times F_{N})  \label{42}
\end{equation}%
of events $F_{1}\in \mathcal{F}_{\Lambda _{1}},...,F_{N}\in \mathcal{F}%
_{\Lambda _{N}}$ observed at the corresponding sites or if outcomes are
real-valued and bounded -- to the expectation value 
\begin{equation}
\left\langle \lambda _{n_{1}}^{(s_{n_{1}})}\cdot ...\cdot \lambda
_{n_{_{M}}}^{(s_{n_{M}})}\right\rangle _{\mathcal{E}_{S,\Lambda
}}=\dint\limits_{\Lambda }\lambda _{n_{1}}\cdot ...\cdot \lambda
_{n_{_{M}}}P_{(s_{1},...,s_{N})}^{(\mathcal{E}_{S,\Lambda })}(\mathrm{d}%
\lambda _{1}\times \cdots \times \mathrm{d}\lambda _{N})  \label{43}
\end{equation}%
of the product $\lambda _{n_{1}}\cdot ...\cdot \lambda _{n_{_{M}}}$ of
outcomes observed at arbitrary $M\leq N$ sites $1\leq n_{1}<...<n_{M}\leq N.$
For $M\geq 2,$ the expectation value (\ref{43}) is referred to (in quantum
information) as \emph{a correlation function. }A correlation function\emph{\ 
}for\emph{\ }an $N$-partite joint measurement is called \emph{\ full }if, in
(\ref{43}),\emph{\ }$M=N.$

If a correlation scenario $\mathcal{E}_{S,\Lambda }$ admits an LHV model,
then every linear combination (\ref{38}) of averages satisfies the tight LHV
constraints \cite{loubenets7}%
\begin{equation}
\mathcal{B}_{\Psi _{S,\Lambda }}^{\inf }\text{ \ }\leq
\dsum\limits_{s_{1},...,s_{_{N}}}\left\langle \psi _{(s_{1},\ldots
,s_{_{N}})}(\lambda _{1},...,\lambda _{N})\right\rangle _{\mathcal{E}%
_{S,\Lambda }}\text{ }|_{_{_{_{LHV}}}}\text{ \ }\leq \mathcal{B}_{\Psi
_{S,\Lambda }}^{\sup },  \label{44}
\end{equation}%
where the LHV constants $\mathcal{B}_{\Psi _{S,\Lambda }}^{\sup }$ and $%
\mathcal{B}_{\Psi _{S,\Lambda }}^{\inf }$ constitute, correspondingly,
supremum and infimum of (\ref{38}) over all LHV scenarios $\mathcal{E}%
_{S,\Lambda }^{lhv}$ and have the form: 
\begin{align}
\mathcal{B}_{\Psi _{S,\Lambda }}^{\sup }& :=\sup_{\mathcal{E}_{S,\Lambda
}^{lhv}}\text{ }\dsum\limits_{s_{1},...,s_{_{N}}}\left\langle \psi
_{(s_{1},\ldots ,s_{_{N}})}(\lambda _{1},...,\lambda _{N})\right\rangle _{%
\mathcal{E}_{S,\Lambda }^{lhv}}  \label{45} \\
& =\sup_{\lambda _{n}^{(s_{n})}\in \Lambda _{n},\forall s_{n},\forall n}%
\text{ }\dsum\limits_{s_{1},...,s_{_{N}}}\psi _{(s_{1},\ldots
,s_{_{N}})}(\lambda _{1}^{(s_{1})},...,\lambda _{N}^{(s_{N})}),  \notag \\
&  \notag \\
\mathcal{B}_{\Psi _{S,\Lambda }}^{\inf }& :=\inf_{\mathcal{E}_{S,\Lambda
}^{lhv}}\text{ }\dsum\limits_{s_{1},...,s_{_{N}}}\left\langle \psi
_{(s_{1},\ldots ,s_{_{N}})}(\lambda _{1},...,\lambda _{N})\right\rangle _{%
\mathcal{E}_{S,\Lambda }^{lhv}}  \notag \\
& =\inf_{\lambda _{n}^{(s_{n})}\in \Lambda _{n},\forall s_{n},\forall n}%
\text{\ }\dsum\limits_{s_{1},...,s_{_{N}}}\psi _{(s_{1},\ldots
,s_{_{N}})}(\lambda _{1}^{(s_{1})},...,\lambda _{N}^{(s_{N})}).  \notag
\end{align}

Constraints (\ref{44}) imply%
\begin{equation}
\left\vert \dsum \limits_{s_{1},...,s_{_{N}}}\left\langle
\psi_{(s_{1},\ldots,s_{_{N}})}(\lambda_{1},...,\lambda _{N})\right\rangle _{%
\mathcal{E}_{S,\Lambda}}\right\vert _{LHV}\leq \mathcal{B}%
_{\Psi_{S,\Lambda}},  \label{47}
\end{equation}
where

\begin{align}
\mathcal{B}_{\Psi_{S,\Lambda}} & :=\max\left\{ \left\vert \mathcal{B}%
_{\Psi_{S,\Lambda}}^{\sup}\right\vert ,\left\vert \mathcal{B}_{\Psi
_{S,\Lambda}}^{\inf}\right\vert \right\}  \label{48} \\
& =\sup_{\lambda_{n}^{(s_{n})}\in\Lambda_{n},\forall s_{n},\forall n}\text{ }%
\left\vert \dsum
\limits_{s_{1},...,s_{_{N}}}\psi_{(s_{1},\ldots,s_{_{N}})}(%
\lambda_{1}^{(s_{1})},...,\lambda_{N}^{(s_{N})})\right\vert .  \notag
\end{align}

Inequalities (\ref{44}) have been introduced in Ref. \cite{loubenets7} and
represent the general form of all unconditional \cite{cond} tight linear LHV
constraints on correlation functions and joint probabilities for an $%
S_{1}\times \cdots \times S_{N}$-setting correlation scenario with outcomes
of an arbitrary type, discrete or continuous.

Note that some of the LHV constraints (\ref{44}) may be fulfilled for a
wider (than LHV) class of correlation scenarios. This is, for example, the
case for those LHV constraints on joint probabilities that follow explicitly
from positivity and nonsignaling of probability distributions $%
P_{(s_{1},...,s_{N})}^{(\mathcal{E}_{S},_{\Lambda})}$ and are, therefore,
fulfilled for any nonsignaling scenario $\mathcal{E}_{S,\Lambda}.$ Moreover,
for some $\Psi_{S,\Lambda}$, the corresponding constraints (\ref{44}) may be
simply trivial -- in the sense that these constraints are fulfilled for each
scenario $\mathcal{E}_{S,\Lambda}$. For example, if we specify (\ref{44})
with functions $\widetilde{\psi}_{(s_{1},\ldots,s_{_{N}})}(\lambda_{1},...,%
\lambda_{N})=1,$ $\forall(\lambda_{1},...,\lambda_{N})\in\Lambda,$ for all
joint measurements $(s_{1},\ldots,s_{_{N}}),$ then 
\begin{equation}
\mathcal{B}_{\widetilde{\Psi}_{S,\Lambda}}^{\inf}=\dsum
\limits_{s_{1},...,s_{_{N}}}\left\langle \text{ }\widetilde{\psi}%
_{(s_{1},\ldots,s_{_{N}})}(\lambda _{1},...,\lambda_{N})\text{ }%
\right\rangle _{\mathcal{E}_{S,\Lambda}}=\mathcal{B}_{\widetilde{\Psi}%
_{S,\Lambda}}^{\sup}=S_{1}\cdot\ldots\cdot S_{N}
\end{equation}
holds for every scenario $\mathcal{E}_{S,\Lambda}.$

If, however, an LHV constraint may be violated in a non-LHV\ case, then it
is generally named after Bell due to his seminal result in Ref. \cite{1}.

\begin{definition}
Each of the tight linear LHV constraints (\ref{44}) that may be violated
under a non-LHV correlation scenario is referred to as a Bell-type
(equivalently, Bell) inequality.
\end{definition}

As it is discussed in section 3 of Ref. \cite{loubenets7}, the general form (%
\ref{44}) covers in a unified manner all unconditional Bell-type
inequalities that were introduced via a variety of methods ever since the
seminal publication of Bell \cite{1}. Note that the original Bell inequality 
\cite{1}, discussed recently in Ref. \cite{loubenets11}, constitutes an
example of conditional Bell-type inequalities.

\section{LqHV modelling of a quantum correlation scenario}

Let, under an $S_{1}\times \cdots \times S_{N}$-setting correlation
scenario, each $N$-partite joint measurement $(s_{1},...,s_{N})$ be
performed on a quantum state $\rho $ on a Hilbert space $\mathcal{H}%
_{1}\otimes \cdots \otimes \mathcal{H}_{N}$ and described by joint
probability measures 
\begin{equation}
\mathrm{tr}[\rho \{\mathrm{M}_{1}^{(s_{1})}(\mathrm{d}\lambda _{1})\otimes
\cdots \otimes \mathrm{M}_{N}^{(s_{_{N}})}(\mathrm{d}\lambda _{N})\}]
\label{1}
\end{equation}%
on the measurable space $(\Lambda _{1}\times \cdots \times \Lambda _{N},$ $%
\mathcal{F}_{\Lambda _{1}}\otimes \cdots \otimes \mathcal{F}_{\Lambda
_{N}}). $ Here, each $\mathrm{M}_{n}^{(s_{n})}$ is a normalized positive
operator-valued (\emph{POV}) measure on a measurable space $(\Lambda _{n},%
\mathcal{F}_{\Lambda _{n}})$ representing on a Hilbert space $\mathcal{H}%
_{n} $ a quantum measurement $s_{n}$ at $n$-th site. For a POV measure $%
\mathrm{M}_{n}^{(s_{n})}$, all its values $\mathrm{M}_{n}^{(s_{n})}(F_{n}),$ 
$F_{n}\in \mathcal{F}_{\Lambda _{n}},$ are positive operators on $\mathcal{H}%
_{n}$ and $\mathrm{M}_{n}^{(s_{n})}(\Lambda _{n})=\mathbb{I}_{\mathcal{H}%
_{n}}.$

We specify this quantum $S_{1}\times\cdots\times S_{N}$-setting correlation
scenario by symbol $\mathcal{E}_{\rho,\mathrm{M}_{S,\Lambda}}$, where 
\begin{equation}
\mathrm{M}_{S,\Lambda}:=\left\{ \mathrm{M}_{n}^{(s_{n})},\text{ }%
s_{n}=1,..,S_{n},\text{ }n=1,...,N\right\}  \label{2}
\end{equation}
is a collection of POV measures describing this quantum scenario and denote
by 
\begin{align}
P_{(s_{1},...,s_{N})}^{(\mathcal{E}_{\rho,\mathrm{M}_{S,\Lambda}})}(\mathrm{d%
}\lambda_{1}\times\cdots\times\mathrm{d}\lambda_{N}) & :=\mathrm{tr}\left[
\rho\left\{ \mathrm{M}_{1}^{(s_{1})}(\mathrm{d}\lambda_{1})\otimes\cdots%
\otimes\mathrm{M}_{N}^{(s_{_{N}})}(\mathrm{d}\lambda_{N})\right\} \right] ,
\label{3} \\
s_{1} & =1,...,S_{1},...,s_{N}=1,...,S_{N},  \notag
\end{align}
its joint probability distributions (\ref{1}).

\begin{theorem}
For every state $\rho $ on a Hilbert space $\mathcal{H}_{1}\otimes \mathcal{%
\cdots }\otimes \mathcal{H}_{N}$ and arbitrary positive integers $%
S_{1},...,S_{N}\geq 1$, each quantum $S_{1}\times ...\times S_{N}$-setting
correlation scenario $\mathcal{E}_{\rho ,\mathrm{M}_{S,\Lambda }}$, with
joint probability distributions (\ref{3}) and outcomes of an arbitrary
spectral type, discrete or continuous, admits an LqHV\ model.
\end{theorem}

\begin{proof}
For a state $\rho $ on $\mathcal{H}_{1}\otimes \mathcal{\cdots }\otimes 
\mathcal{H}_{N}$, let $T_{S_{1}\times \cdots \times S_{N}}^{(\rho )}$ be an $%
S_{1}\times ...\times S_{N}$-setting source operator on space $\mathcal{H}%
_{1}^{\otimes S_{1}}\otimes \cdots \otimes \mathcal{H}_{N}^{\otimes S_{N}}$,
see definition 1 in section 2. For each scenario $\mathcal{E}_{\rho ,\mathrm{%
M}_{S,\Lambda }},$ the normalized real-valued measure%
\begin{align}
& \mu _{T_{S_{1}\times \cdots \times S_{N}}^{(\rho )}}^{(\rho ,\mathrm{M}%
_{S,\Lambda })}\left( \mathrm{d}\lambda _{1}^{(1)}\times \cdots \times 
\mathrm{d}\lambda _{1}^{(S_{1})}\times \cdots \times \mathrm{d}\lambda
_{N}^{(1)}\times \cdots \times \mathrm{d}\lambda _{N}^{(S_{N})}\right)
\label{22} \\
& :=\mathrm{tr}[T_{S_{1}\times \cdots \times S_{N}}^{(\rho )}\text{ }\{%
\mathrm{M}_{1}^{(1)}(\mathrm{d}\lambda _{1}^{(1)})\otimes \cdots \otimes 
\mathrm{M}_{1}^{(S_{1})}(\mathrm{d}\lambda _{1}^{(S_{1})})  \notag \\
& \otimes \cdots \otimes \mathrm{M}_{N}^{(1)}(\mathrm{d}\lambda
_{N}^{(1)})\otimes \cdots \otimes \mathrm{M}_{N}^{(S_{N})}(\mathrm{d}\lambda
_{N}^{(S_{N})})\}]  \notag
\end{align}%
on the direct product space $(\Lambda _{1}^{S_{1}}\times \cdots \times
\Lambda _{N}^{S_{N}},\mathcal{F}_{\Lambda _{1}}^{\otimes S_{1}}\otimes
\cdots \otimes \mathcal{F}_{\Lambda _{N}}^{\otimes S_{N}})$ returns all
joint probability distributions $P_{(s_{1},...,s_{N})}^{(\mathcal{E}_{\rho ,%
\mathrm{M}_{S,\Lambda }})}$ of scenario $\mathcal{E}_{\rho ,\mathrm{M}%
_{S,\Lambda }}$ as the corresponding marginals. Due to bound (\ref{B10})
proved in appendix B and relation (\ref{19}), the total variation norm of
measure (\ref{22}) is upper bounded by $||T_{S_{1}\times \cdots \times
S_{N}}^{(\rho )}||_{1}<\infty .$ This and relation (\ref{var}) imply that
the normalized real-valued measure (\ref{22}) is bounded. Thus, for each
quantum scenario $\mathcal{E}_{\rho ,\mathrm{M}_{S,\Lambda }},$ the
constructed measure $\mu _{T_{S_{1}\times \cdots \times S_{N}}^{(\rho
)}}^{(\rho ,\mathrm{M}_{S,\Lambda })}$ satisfies the sufficiency condition
of theorem 1 on LqHV modelling. This proves the statement.\medskip
\end{proof}

If, for a state $\rho ,$ every quantum scenario $\mathcal{E}_{\rho ,\mathrm{M%
}_{S,\Lambda }}$ (i.e. for an arbitrary collection $\mathrm{M}_{S,\Lambda }$
of POV measures and an arbitrary outcome set $\Lambda $) admits an LHV
model, then, according to our terminology in Ref. \cite{loubenets3}, this
state $\rho $ admits the $S_{1}\times \cdots \times S_{N}$-setting LHV
description. In the latter case, state $\rho $ admits \cite{loubenets3} an $%
L_{1}\times \cdots \times L_{N}$-setting LHV description for all $L_{1}\leq
S_{1},...,L_{N}\leq S_{N},$ but does not need to admit the LHV description
whenever at least one $L_{n}>S_{n}.$

Via a similar terminology for the LqHV case, theorem 2 reads -- \emph{every} 
$N$\emph{-partite quantum state }$\rho $\emph{\ admits an }$S_{1}\times
\cdots \times S_{N}$\emph{-setting LqHV description for arbitrary numbers }$%
S_{1},...,S_{N}$\emph{\ of\ measurements at }$N$\emph{\ sites.}

In view of theorems 1, 2, corollary 1 and relation \ref{35}, let us
introduce, for a quantum correlation scenario $\mathcal{E}_{\rho ,\mathrm{M}%
_{S,\Lambda }}$, the parameter 
\begin{equation}
\mathrm{\gamma }_{\mathcal{E}_{\rho ,\mathrm{M}_{S,\Lambda }}}:=\inf
\left\Vert \mu _{\mathcal{E}_{\rho ,\mathrm{M}_{S,\Lambda }}}\right\Vert
_{var}\geq 1,  \label{q1}
\end{equation}%
where, infimum is taken over all normalized bounded real-valued measures $%
\mu _{\mathcal{E}_{\rho ,\mathrm{M}_{S,\Lambda }}}$, each returning all
distributions $P_{(s_{1},...,s_{N})}^{(\mathcal{E}_{\rho ,\mathrm{M}%
_{S,\Lambda }})}$ of scenario $\mathcal{E}_{\rho ,\mathrm{M}_{S,\Lambda }}$
as the corresponding marginals.

The following lemma is proved in appendix B.

\begin{lemma}
A quantum correlation scenario $\mathcal{E}_{\rho ,\mathrm{M}_{S,\Lambda }}$
admits an LHV model if and only if $\mathrm{\gamma }_{\mathcal{E}_{\rho ,%
\mathrm{M}_{S,\Lambda }}}=1.\medskip $
\end{lemma}

Introduce also the state parameters%
\begin{align}
\mathrm{\Upsilon}_{S_{1}\times\cdots\times S_{N}}^{(\rho,\Lambda)} & :=\sup_{%
\mathrm{M}_{S,\Lambda}}\mathrm{\gamma}_{\mathcal{E}_{\rho ,\mathrm{M}%
_{S,\Lambda}}}\geq1,  \label{s1} \\
\mathrm{\Upsilon}_{S_{1}\times\cdots\times S_{N}}^{(\rho)} & :=\sup_{\Lambda
}\mathrm{\Upsilon}_{S_{1}\times\cdots\times S_{N}}^{(\rho,\Lambda)}\geq1.
\label{s}
\end{align}

\begin{proposition}
(a)\textrm{\ }For a state $\rho$ on $\mathcal{H}_{1}\otimes\mathcal{\cdots }%
\otimes\mathcal{H}_{N},$ each quantum $L_{1}\times\cdots\times L_{N}$%
-setting scenario $\mathcal{E}_{\rho,\mathrm{M}_{L,\Lambda}},$ with $%
L_{1}\leq S_{1},...,L_{N}\leq S_{N}$ and an outcome set $\Lambda,$ admits an
LHV model if and only if $\mathrm{\Upsilon}_{S_{1}\times\cdots\times
S_{N}}^{(\rho,\Lambda)}=1;\medskip$\newline
(b) A state $\rho$ on $\mathcal{H}_{1}\otimes\mathcal{\cdots}\otimes\mathcal{%
H}_{N}$ admits an $S_{1}\times\cdots\times S_{N}$-setting LHV description if
and only if $\mathrm{\Upsilon}_{S_{1}\times\cdots\times S_{N}}^{(\rho)}=1.$
\end{proposition}

\begin{proof}
If each scenario $\mathcal{E}_{\rho ,\mathrm{M}_{L,\Lambda }}$, where $%
L_{1}\leq S_{1},...,L_{N}\leq S_{N},$ admits an LHV\ model, then, by lemma
2, $\mathrm{\gamma }_{\mathcal{E}_{\rho ,\mathrm{M}_{L,\Lambda }}}=1$ for
all collections $\mathrm{M}_{L,\Lambda }$ of POV measures, where $L_{1}\leq
S_{1},...,L_{N}\leq S_{N}.$ Hence, due to its definition (\ref{s1}),
parameter $\mathrm{\Upsilon }_{S_{1}\times \cdots \times S_{N}}^{(\rho
,\Lambda )}=1.$ Conversely, let $\mathrm{\Upsilon }_{S_{1}\times \cdots
\times S_{N}}^{(\rho ,\Lambda )}=1.$ Then, in view of (\ref{s1}), (\ref{q1}%
), $\mathrm{\gamma }_{\mathcal{E}_{\rho ,\mathrm{M}_{S,\Lambda }}}=1$ for
all collections $\mathrm{M}_{S,\Lambda }$ of POV\ measures. By lemma 2, this
implies LHV modelling of every quantum correlation scenario $\mathcal{E}%
_{\rho ,\mathrm{M}_{S,\Lambda }}.$ By proposition 3 in Ref. \cite{loubenets3}%
, the latter, in turn, implies LHV modelling of each quantum correlation
scenario $\mathcal{E}_{\rho ,\mathrm{M}_{L,\Lambda }}$ with settings $%
L_{1}\leq S_{1},...,L_{N}\leq S_{N}.$ This proves the sufficiency part of
statement (a). Statement (b) is proved quite similarly. "
\end{proof}

Thus, for an $N$-partite state $\rho ,$ it is specifically the state
parameter $\Upsilon _{S_{1}\times \cdots \times S_{N}}^{(\rho )}$ that
determines quantitatively a possibility of an LHV description of all quantum
scenarios (\ref{3}) with settings up to setting $S_{1}\times \cdots \times
S_{N}$ and outcomes of an arbitrary type.

\begin{proposition}
For a state $\rho$ on a Hilbert space $\mathcal{H}_{1}\otimes\mathcal{\cdots 
}\otimes\mathcal{H}_{N}$ and arbitrary positive integers $%
S_{1},...,S_{N}\geq1,$ 
\begin{equation}
1\leq\mathrm{\Upsilon}_{S_{1}\times\cdots\times S_{N}}^{(\rho,\Lambda)}\leq%
\mathrm{\Upsilon}_{S_{1}\times\cdots\times S_{N}}^{(\rho)}\leq
\inf_{T_{S_{1}\times\cdots\times\underset{\overset{\uparrow}{n}}{1}%
\times\cdots\times S_{N}}^{(\rho)},\text{ }\forall n}\text{ }%
||T_{S_{1}\times\cdots\times\underset{\overset{\uparrow}{n}}{1}%
\times\cdots\times S_{N}}^{(\rho)}\text{ }||_{cov},  \label{x}
\end{equation}
where (i) infimum is taken over all source operators $T_{S_{1}\times
\cdots\times\underset{\overset{\uparrow}{n}}{1}\times\cdots\times
S_{N}}^{(\rho)}$ with only one setting at $n$-th site and over all $%
n=1,...,N;$ (ii) $\left\Vert \cdot\right\Vert _{cov}$ is the covering norm
(see definition 4 in section 2).
\end{proposition}

\begin{proof}
Inequalities (\ref{x}) follow from (\ref{s1}), (\ref{s}) and the upper bound 
\begin{equation}
\mathrm{\gamma}_{\mathcal{E}_{\rho,\mathrm{M}_{S,\Lambda}}}\leq\inf
_{T_{S_{1}\times\cdots\times\underset{\overset{\uparrow}{n}}{1}\times
\cdots\times S_{N}}^{(\rho)},\text{ }\forall n}\text{ }||T_{S_{1}\times
\cdots\times\underset{\overset{\uparrow}{n}}{1}\times\cdots\times
S_{N}}^{(\rho)}\text{ }||_{cov}  \label{new}
\end{equation}
constituting relation\ (\ref{B18}) of lemma 5 in appendix B.
\end{proof}

Propositions 3, 4 imply the following general statements on an $%
S_{1}\times\cdots\times S_{N}$-setting LHV description of an $N$-partite
quantum state.

\begin{proposition}
(a)\textrm{\ }Every\textrm{\ }$N$-partite\textrm{\ }quantum state $\rho$
admits an $1\times\cdots\times1\times S_{n}\times1\times\cdots\times1$%
-setting LHV description; \medskip\newline
(b)\textrm{\ }If, for a state $\rho$ on a Hilbert space $\mathcal{H}%
_{1}\otimes\cdots\otimes\mathcal{H}_{N},$ there exists a tensor positive
source operator $T_{S_{1}\times\cdots\times \underset{\overset{\uparrow}{n}}{%
1}\times\cdots\times S_{N}}^{(\rho)}$ for some $n$, then $\rho$ admits the $%
S_{1}\times\cdots\times\widetilde{S}_{n}\times\cdots\times S_{N}$-setting
LHV description for an arbitrary number $\widetilde{S}_{n}\geq1$ of settings
at this $n$-th site;\medskip\newline
(c) If, for a state $\rho$ on a Hilbert space $\mathcal{H}_{1}\otimes\cdots
\otimes\mathcal{H}_{N},$ there exists a tensor positive source operator $%
T_{S_{1}\times\cdots\times S_{N}}^{(\rho)}$, then $\rho$ admits the $%
S_{1}\times\cdots\times S_{n}^{\prime}\times\cdots\times S_{N}$-setting LHV
description for an arbitrary number $S_{n}^{\prime}$ of measurements at
every $n$-th site.
\end{proposition}

\begin{proof}
Since $T_{1\times\cdots\times1}^{(\rho)}=\rho$ and $||\rho||_{cov}=1,$ from
bound (\ref{x}) it follows that $\mathrm{\Upsilon}_{1\times\cdots\times
S_{n}\times\cdots\times1}^{(\rho)}=1$ for every $N$-partite state $\rho$. By
proposition 3, this proves statement \textrm{(a)}.

If, for an $N$-partite state $\rho,$ there exists a tensor positive source
operator $T_{S_{1}\times\cdots\times\underset{\overset{\uparrow}{n}}{1}%
\times\cdots\times S_{N}}^{(\rho)}\overset{\otimes}{\geq}0$ for some $n, $
then from relation (\ref{19'}) it follows that, for this source operator,
the covering norm $||T_{S_{1}\times\cdots\times\underset{\overset{\uparrow}{n%
}}{1}\times\cdots\times S_{N}}^{(\rho)}||_{cov}$ $=1. $ In view of bound (%
\ref{x}), the latter implies $\Upsilon_{S_{1}\times\cdots\times\widetilde {S}%
_{n}\times\cdots\times S_{N}}^{(\rho)}=1$ for any number $\widetilde{S}%
_{n}\geq1$ of settings at this $n$-th site. By proposition 3, this proves
statement \textrm{(b)}.

Let an $N$-partite state $\rho $ have a tensor positive source operator $%
T_{S_{1}\times \cdots \times S_{N}}^{(\rho )}\overset{\otimes }{\geq }0$.
Then, for each $n=1,...,N,$ the operator on $\mathcal{H}_{1}^{(S_{1})}%
\otimes \cdots \otimes \mathcal{H}_{n}$ $\otimes \cdots \otimes \mathcal{H}%
_{N}^{(S_{N})}$ reduced from $T_{S_{1}\times \cdots \times S_{N}}^{(\rho )}$
constitutes a tensor positive source operator $T_{S_{1}\times \cdots \times 
\underset{\overset{\uparrow }{n}}{1}\times \cdots \times S_{N}}^{(\rho )}$
for state $\rho $ and, therefore, statement \textrm{(c)} follows from
statement \textrm{(b).}
\end{proof}

Statement \textrm{(a)} of proposition 5 agrees with proposition 2 of Ref. 
\cite{loubenets3} on the LHV description of a general correlation scenario
with setting $S\times1\times\cdots\times1.$

Specified for a bipartite case $(N=2)$, statements \textrm{(b),} \textrm{(c)}
of proposition 5 are consistent in view of note \cite{ter} with theorems 1,
2 in Ref. \cite{terhal}.

\section{Quantum violations of Bell-type inequalities}

Consider a linear combination (\ref{38}) of averages (\ref{av}), arising
under a quantum $S_{1}\times \cdots \times S_{N}$-setting correlation
scenario $\mathcal{E}_{\rho ,\text{ }\mathrm{M}_{S,\Lambda }}$ and specified
by a family $\Psi _{S,\Lambda }=\{\psi _{(s_{1},\ldots ,s_{_{N}})}\}$ of
bounded measurable real-valued functions $\psi _{(s_{1},\ldots ,s_{_{N}})}:$ 
$\Lambda _{1}\times \cdots \times \Lambda _{N}\rightarrow \mathbb{R}.$

By theorem 2, every quantum scenario $\mathcal{E}_{\rho ,\text{ }\mathrm{M}%
_{S,\Lambda }}$ admits an LqHV model and, by theorem 1, the latter is
equivalent to the existence of a bounded real-valued measure $\mu _{\mathcal{%
E}_{\rho ,\mathrm{M}_{S,\Lambda }}}$ returning all joint probability
distributions $P_{(s_{1},...,s_{N})}^{(\mathcal{E}_{\rho ,\text{ }\mathrm{M}%
_{S,\Lambda }})}$ of scenario $\mathcal{E}_{\rho ,\text{ }\mathrm{M}%
_{S,\Lambda }}$ as the corresponding marginals. \ Therefore, for a quantum
scenario $\mathcal{E}_{\rho ,\text{ }\mathrm{M}_{S,\Lambda }},$ a linear
combination (\ref{38}) of averages takes the form 
\begin{align}
& \dsum\limits_{s_{1},...,s_{N}}\left\langle \psi _{(s_{1},\ldots
,s_{_{N}})}(\lambda _{1},...,\lambda _{N})\right\rangle _{\mathcal{E}_{\rho ,%
\mathrm{M}_{S,\Lambda }}}  \label{49} \\
& =\dint \text{ }\dsum\limits_{s_{1},...,s_{N}}\psi _{(s_{1},\ldots
,s_{_{N}})}(\lambda _{1}^{(s_{1})},...,\lambda _{N}^{(s_{N})})\text{ }\mu
_{_{\mathcal{E}_{\rho ,\mathrm{M}_{S,\Lambda }}}}(\mathrm{d}\lambda
_{1}^{(1)}\times \cdots \times \mathrm{d}\lambda _{1}^{(S_{1})}  \notag \\
& \times \cdots \times \mathrm{d}\lambda _{N}^{(1)}\times \cdots \times 
\mathrm{d}\lambda _{N}^{(S_{N})}).  \notag
\end{align}%
Substituting the Jordan decomposition (see remark 2) of measure $\mu _{_{%
\mathcal{E}_{\rho ,\mathrm{M}_{S,\Lambda }}}}$ into (\ref{49}) and taking
into the account (\ref{35''}), (\ref{35}), we derive:%
\begin{align}
& \mathcal{B}_{\Psi _{S,\Lambda }}^{\inf }-\frac{\left\Vert \mu _{_{\mathcal{%
E}_{\rho ,\mathrm{M}_{S,\Lambda }}}}\right\Vert _{var}-1}{2}(\mathcal{B}%
_{\Psi _{S,\Lambda }}^{\sup }-\mathcal{B}_{\Psi _{S,\Lambda }}^{\inf })
\label{50} \\
& \leq \dsum\limits_{s_{1},...,s_{N}}\left\langle \psi _{(s_{1},\ldots
,s_{_{N}})}(\lambda _{1},...,\lambda _{N})\text{ }\right\rangle _{\mathcal{E}%
_{\rho ,\mathrm{M}_{S,\Lambda }}}  \notag \\
& \leq \mathcal{B}_{\Psi _{S,\Lambda }}^{\sup }+\frac{\left\Vert \mu _{_{%
\mathcal{E}_{\rho ,\mathrm{M}_{S,\Lambda }}}}\right\Vert _{var}-1}{2}(%
\mathcal{B}_{\Psi _{S,\Lambda }}^{\sup }-\mathcal{B}_{\Psi _{S,\Lambda
}}^{\inf }),  \notag
\end{align}%
where $\mathcal{B}_{\Psi _{S,\Lambda }}^{\sup },$ $\mathcal{B}_{\Psi
_{S,\Lambda }}^{\inf }$ are the LHV constants (\ref{45}) and $\left\Vert \mu
_{\mathcal{E}_{\rho ,\mathrm{M}_{S,\Lambda }}}\right\Vert _{var}\geq 1$ is
the total variation norm of measure $\mu _{\mathcal{E}_{\rho ,\text{ }%
\mathrm{M}_{S,\Lambda }}}$.

Since inequalities (\ref{50}) hold for each measure $\mu_{\mathcal{E}_{\rho,%
\text{ }\mathrm{M}_{S,\Lambda}}}$ returning all joint probability
distributions $P_{(s_{1},...,s_{N})}^{(\mathcal{E}_{\rho,\text{ }\mathrm{M}%
_{S,\Lambda}})}$ of scenario $\mathcal{E}_{\rho,\text{ }\mathrm{M}%
_{S,\Lambda}}$ as the corresponding marginals, we have:%
\begin{align}
& \mathcal{B}_{\Psi_{S,\Lambda}}^{\inf}-\frac{\mathrm{\gamma}_{\mathcal{E}%
_{\rho,\mathrm{M}_{S,\Lambda}}}-1}{2}\left( \mathcal{B}_{\Psi_{S,\Lambda}}^{%
\sup}-\mathcal{B}_{\Psi_{S,\Lambda}}^{\inf}\right)  \label{52} \\
& \leq\dsum \limits_{s_{1},...,s_{N}}\left\langle
\psi_{(s_{1},\ldots,s_{_{N}})}(\lambda_{1},...,\lambda_{N})\text{ }%
\right\rangle _{\mathcal{E}_{\rho,\mathrm{M}_{S,\Lambda}}}  \notag \\
& \leq\mathcal{B}_{\Psi_{S,\Lambda}}^{\sup}+\frac{\mathrm{\gamma }_{\mathcal{%
E}_{\rho,\mathrm{M}_{S,\Lambda}}}-1}{2}\left( \mathcal{B}_{\Psi_{S,%
\Lambda}}^{\sup}-\mathcal{B}_{\Psi_{S,\Lambda}}^{\inf}\right) ,  \notag
\end{align}
where $\mathrm{\gamma}_{\mathcal{E}_{\rho,\mathrm{M}_{S,\Lambda}}}=\inf
_{\mu_{_{_{\mathcal{E}_{\rho,\mathrm{M}_{S,\Lambda}}}}}}\left\Vert \mu_{%
\mathcal{E}_{\rho,\mathrm{M}_{S,\Lambda}}}\right\Vert _{var}\geq1$ is the
scenario parameter (\ref{q1}).

Maximizing (\ref{52}) over all possible scenarios $\mathcal{E}_{\rho,\text{ }%
\mathrm{M}_{S,\Lambda}}$, performed on a quantum state $\rho$ and with
outcomes in a set $\Lambda,$ and taking into the account that $\sup _{%
\mathrm{M}_{S,\Lambda}}$ $\mathrm{\gamma}_{\mathcal{E}_{\rho,\mathrm{M}%
_{S,\Lambda}}}=\mathrm{\Upsilon}_{S_{1}\times\cdots\times
S_{N}}^{(\rho,\Lambda)}$ is the state parameter (\ref{s1}), for an $N$%
-partite quantum state $\rho$ and a function collection $\Psi_{S,\Lambda}=\{%
\psi_{(s_{1},\ldots,s_{_{N}})}\},$ we derive the following analogs 
\begin{align}
& \mathcal{B}_{\Psi_{S,\Lambda}}^{\inf}-\frac{\mathrm{\Upsilon}%
_{S_{1}\times\cdots\times S_{N}}^{(\rho,\Lambda)}-1}{2}(\mathcal{B}%
_{\Psi_{S,\Lambda }}^{\sup}-\mathcal{B}_{\Psi_{S,\Lambda}}^{\inf})
\label{52'} \\
& \leq\dsum \limits_{s_{1},...,s_{N}}\left\langle
\psi_{(s_{1},\ldots,s_{_{N}})}(\lambda_{1},...,\lambda_{N})\text{ }%
\right\rangle _{\mathcal{E}_{\rho,\mathrm{M}_{S,\Lambda}}}  \notag \\
& \leq\mathcal{B}_{\Psi_{S,\Lambda}}^{\sup}+\frac{\mathrm{\Upsilon}%
_{S_{1}\times\cdots\times S_{N}}^{(\rho,\Lambda)}-1}{2}(\mathcal{B}%
_{\Psi_{S,\Lambda}}^{\sup}-\mathcal{B}_{\Psi_{S,\Lambda}}^{\inf})  \notag
\end{align}
of the LHV constraints (\ref{44}). Since inequalities (\ref{52'}) are
non-trivial only for those $\Psi_{S,\Lambda}$ that correspond via (\ref{44})
to Bell-type inequalities, we refer to (\ref{52'}) as \emph{the} \emph{%
analogs of Bell-type inequalities for an }$N$\emph{-partite quantum state }$%
\rho.$

From (\ref{52'}) it follows%
\begin{equation}
\left\vert \dsum \limits_{s_{1},...,s_{_{N}}}\left\langle
\psi_{(s_{1},\ldots,s_{_{N}})}(\lambda_{1},...,\lambda _{N})\right\rangle _{%
\mathcal{E}_{\rho,\mathrm{M}_{S,\Lambda}}}\right\vert \text{ }\leq\mathrm{%
\Upsilon}_{S_{1}\times\cdots\times S_{N}}^{(\rho,\Lambda )}\text{ }\mathcal{B%
}_{\Psi_{S,\Lambda}},  \label{52'_}
\end{equation}
where $\mathcal{B}_{\Psi_{S,\Lambda}}$ is the LHV constant (\ref{48}).

\begin{remark}
The quantum constraints (\ref{52'}), (\ref{52'_}) are equivalent iff $%
\mathcal{B}_{\Psi_{S,\Lambda}}^{\inf}=-\mathcal{B}_{\Psi_{S,\Lambda}}^{\sup}$%
. For an arbitrary function collection $\Psi_{S,\Lambda},$ (\ref{52'}) $%
\Rightarrow$ (\ref{52'_}) but not vice versa. In order to see a difference
between these two types of quantum constraints for an arbitrary $\Psi
_{S,\Lambda}$, let us take $\widetilde{\Psi}_{S,\Lambda}$ for which $%
\mathcal{B}_{\widetilde{\Psi}_{S,\Lambda}}^{\sup}=0.$ In this case, the
left-hand and the right-hand sides of (\ref{52'}) are equal to 
\begin{equation}
-\frac{\mathrm{\Upsilon}_{S_{1}\times\cdots\times S_{N}}^{(\rho,\Lambda)}+1}{%
2}|\mathcal{B}_{\Psi_{S,\Lambda}}^{\inf}|,\text{ \ \ \ }\frac {\mathrm{%
\Upsilon}_{S_{1}\times\cdots\times S_{N}}^{(\rho,\Lambda)}-1}{2}|\mathcal{B}%
_{\Psi_{S,\Lambda}}^{\inf}|,
\end{equation}
respectively, whereas the right hand side of (\ref{52'_}) is given by $%
\Upsilon_{S_{1}\times\cdots\times S_{N}}^{(\rho,\Lambda)}|\mathcal{B}%
_{\Psi_{S,\Lambda}}^{\inf}|.$ Note that, specifically for bipartite
Bell-type inequalities with $\mathcal{B}_{\widetilde{\Psi}%
_{S,\Lambda}}^{\sup}=0,$ the maximal violations by two-qubit states have
been analyzed numerically in Ref. \cite{brunner}.
\end{remark}

The following statement (proved in appendix C) shows that the quantum
constraints (\ref{52'}), (\ref{52'_}) are \emph{tight} in the sense that the
state parameter $\Upsilon_{S_{1}\times\cdots\times S_{N}}^{(\rho,\Lambda)}$
represents the maximal violation by state $\rho$ of all Bell-type
inequalities (either on correlation functions or on joint probabilities) for
a given outcome set $\Lambda$ and settings $L_{1}\times\cdots\times L_{N}$
with $L_{1}\leq S_{1},...,L_{N}\leq S_{N}$.

\begin{lemma}
In (\ref{52'_}), parameter $\mathrm{\Upsilon }_{S_{1}\times \cdots \times
S_{N}}^{(\rho ,\Lambda )}=\sup_{\mathrm{M}_{S,\Lambda }}\mathrm{\gamma }_{%
\mathcal{E}_{\rho ,\mathrm{M}_{S,\Lambda }}}$ is otherwise expressed by%
\begin{equation}
\mathrm{\Upsilon }_{S_{1}\times \cdots \times S_{N}}^{(\rho ,\Lambda
)}=\sup_{_{\substack{ \mathrm{M}_{S,\Lambda },\Psi _{S,\Lambda },  \\ 
\mathcal{B}_{\Psi _{S,\Lambda }}\neq 0}}}\left\vert \frac{1}{\mathcal{B}%
_{\Psi _{S,\Lambda }}}\dsum\limits_{s_{1},...,s_{_{N}}}\left\langle \psi
_{(s_{1},\ldots ,s_{_{N}})}(\lambda _{1},...,\lambda _{N})\right\rangle _{%
\mathcal{E}_{\rho ,\text{ }\mathrm{M}_{S,\Lambda }}}\right\vert ,  \label{64}
\end{equation}%
where supremum is taken over all non-trivial $(\mathcal{B}_{\Psi _{S,\Lambda
}}\neq 0)$ families $\Psi _{S,\Lambda }=\{\psi _{(s_{1},\ldots ,s_{_{N}})}\}$
of bounded measurable real-valued functions on $\Lambda =\Lambda _{1}\times
\cdots \times \Lambda _{N}$ and over all possible families $\mathrm{M}%
_{S,\Lambda }=\{\mathrm{M}_{n}^{(s_{n})}\}$ of POV measures on spaces $%
(\Lambda _{n},\mathcal{F}_{\Lambda _{n}})$.$\medskip $
\end{lemma}

From lemma 3 it follows that the state parameter $\mathrm{\Upsilon }%
_{S_{1}\times \cdots \times S_{N}}^{(\rho )}:=\sup_{\Lambda }\mathrm{%
\Upsilon }_{S_{1}\times \cdots \times S_{N}}^{(\rho ,\Lambda )},$ introduced
by relation (\ref{s}) and discussed in proposition 3, is otherwise expressed
by 
\begin{equation}
\mathrm{\Upsilon }_{S_{1}\times \cdots \times S_{N}}^{(\rho )}=\sup_{ 
_{\substack{ \Lambda ,\text{ }\mathrm{M}_{S,\Lambda },\Psi _{S,\Lambda }, 
\\ \mathcal{B}_{\Psi _{S,\Lambda }}\neq 0}}}\left\vert \frac{1}{\mathcal{B}%
_{\Psi _{S,\Lambda }}}\dsum\limits_{s_{1},...,s_{_{N}}}\left\langle \psi
_{(s_{1},\ldots ,s_{_{N}})}(\lambda _{1},...,\lambda _{N})\right\rangle _{%
\mathcal{E}_{\rho ,\text{ }\mathrm{M}_{S,\Lambda }}}\right\vert  \label{t}
\end{equation}%
and, therefore, represents the maximal violation by state $\rho $ of all
Bell-type inequalities on correlation functions and joint probabilities for
settings up to setting $S_{1}\times \cdots \times S_{N}$ and an arbitrary
outcome set $\Lambda .$

\begin{proposition}
For a state $\rho$ on a Hilbert space $\mathcal{H}_{1}\otimes\mathcal{\cdots
\otimes H}_{N},$ the following statements are mutually equivalent:\newline
\textrm{(a)} State $\rho$ admits the $S_{1}\times\cdots\times S_{N}$-setting
LHV description;\newline
\textrm{(b)} Parameter $\mathrm{\Upsilon}_{S_{1}\times\cdots\times
S_{N}}^{(\rho)}=1;$\newline
\textrm{(c)} State $\rho$ does not violate any Bell-type inequality with
settings $L_{1}\leq S_{1},...,L_{N}\leq S_{N}$ and outcomes in an arbitrary
set $\Lambda.$
\end{proposition}

\begin{proof}
Equivalence \textrm{(a)}$\Leftrightarrow$\textrm{(b)} follows from
proposition 3. Implication\textrm{\ (a)} $\Rightarrow$ \textrm{(c)} follows
from definition 6 of a Bell-type inequality. Let \textrm{(c)} hold. Then
from (\ref{t}) it follows $\Upsilon_{S_{1}\times\cdots\times
S_{N}}^{(\rho)}=1,$ so that $\mathrm{(c)}\Rightarrow\mathrm{(b)}$. Thus, we
have proved $\mathrm{(a)}\Leftrightarrow\mathrm{(b),}$ \textrm{(a)} $%
\Rightarrow$ \textrm{(c), }$\mathrm{(c)}\Rightarrow\mathrm{(b).}$ These
implications prove the mutual equivalence of statements \textrm{(a)}, 
\textrm{(b), (c)}.\medskip
\end{proof}

The following theorem introduces a general analytical upper bound on the
maximal violation $\mathrm{\Upsilon}_{S_{1}\times\cdots\times
S_{N}}^{(\rho)} $ by state $\rho$ of all $S_{1}\times\cdots\times S_{N}$%
\emph{-}setting\emph{\ }Bell-type inequalities -- \emph{the maximal }$%
S_{1}\times\cdots\times S_{N}$\emph{-setting Bell} \emph{violation for state 
}$\rho,$ for short.

\begin{theorem}
For every quantum state $\rho$ on a Hilbert space $\mathcal{H}_{1}\otimes%
\mathcal{\cdots\otimes H}_{N}$ and arbitrary positive integers $%
S_{1},...,S_{N}\geq1,$ the maximal $S_{1}\times\cdots\times S_{N}$\emph{-}%
setting Bell violation\emph{\ }$\mathrm{\Upsilon}_{S_{1}\times \cdots\times
S_{N}}^{(\rho)}\geq1$ is upper bounded by 
\begin{align}
\mathrm{\Upsilon}_{S_{1}\times\cdots\times S_{N}}^{(\rho)} & \leq
\inf_{T_{S_{1}\times\cdots\times\underset{\overset{\uparrow}{n}}{1}%
\times\cdots\times S_{N}}^{(\rho)},\text{ }\forall n}\text{ }%
||T_{S_{1}\times\cdots\times\underset{\overset{\uparrow}{n}}{1}%
\times\cdots\times S_{N}}^{(\rho)}\text{ }||_{cov}  \label{67} \\
& \leq\inf_{T_{S_{1}\times\cdots\times\underset{\overset{\uparrow}{n}}{1}%
\times\cdots\times S_{N}}^{(\rho)},\text{ }\forall n}\text{ }%
||T_{S_{1}\times\cdots\times\underset{\overset{\uparrow}{n}}{1}\times
\cdots\times S_{N}}^{(\rho)}\text{ }||_{1},  \notag
\end{align}
where infimum is taken over all source operators $T_{S_{1}\times\cdots \times%
\underset{\overset{\uparrow}{n}}{1}\times\cdots\times S_{N}}^{(\rho)}$ with
only one setting at $n$-th site and over all $n=1,...,N$ and $\left\Vert
\cdot\right\Vert _{cov}$, $\left\Vert \cdot\right\Vert _{1}$ mean the
covering norm and the trace norm, respectively.
\end{theorem}

\begin{proof}
The statement follows from relation (\ref{t}), proposition 4 and bound (\ref%
{19}).
\end{proof}

\subsection{Numerical estimates}

In this section, via the analytical upper bound (\ref{67}) we estimate the
maximal $S_{1}\times\cdots\times S_{N}$-setting Bell violation $\mathrm{%
\Upsilon}_{S_{1}\times\cdots\times S_{N}}^{(\rho)}$ in terms of numerical
characteristics of quantum correlation scenarios such as a number $N\geq2$
of sites, a number $S_{n}\geq1$ of measurements and the Hilbert space
dimension $d_{n}:=\dim\mathcal{H}_{n}$ at each $n$-th of $N$ sites.

Let us first evaluate $\mathrm{\Upsilon}_{S_{1}\times\cdots\times
S_{N}}^{(\rho)}$ for some concrete quantum states generally used in quantum
information processing.

For the two-qubit singlet $\psi_{singlet}=\frac{1}{\sqrt{2}}(e_{1}\otimes
e_{2}-e_{2}\otimes e_{1}),$ the analytical upper bound (\ref{67}) and
relation (\ref{A23}) imply 
\begin{equation}
\mathrm{\Upsilon}_{S\times2}^{(\rho_{singlet})}\leq\sqrt{3},\text{ \ }%
\forall S\geq2.  \label{sing}
\end{equation}
Note that, due to Tsirelson's bound \cite{tsirelson} and the analysis of
Fine \cite{fine}, violation by a bipartite state $\rho$ of an arbitrary $%
2\times 2$-setting Bell-type inequality (either on correlation functions or
on joint probabilities) for two settings and two outcomes per site cannot
exceed $\sqrt{2}$ -- in our notation $\mathrm{\Upsilon}_{2\times2}^{(\rho
,\{\lambda_{1},\lambda_{2}\}^{2})}\leq\sqrt{2}$. The maximal violation by
the two-qubit singlet of all \emph{correlation} Bell-type inequalities is
given \cite{acin} by the Grothendieck's constant $K_{G}(3)$ of order 3 and
it is known \cite{Finch, krivine} that $\sqrt{2}\leq K_{G}(3)\leq1.5164.$

Consider also the maximal Bell violations for the N-qudit Greenberger- Horne
- Zeilinger (GHZ) state 
\begin{equation}
\psi _{d}=\frac{1}{\sqrt{d}}\sum_{j=1}^{d}|j\rangle ^{\otimes N}\in \left( 
\mathbb{C}^{d}\right) ^{\otimes N}  \label{74}
\end{equation}%
and the generalized $N$-qubit GHZ state 
\begin{equation}
\psi _{2}^{(gen)}=\sin \varphi \text{ }|1\rangle ^{\otimes N}+\cos \varphi 
\text{ }|2\rangle ^{\otimes N}\in \left( \mathbb{C}^{2}\right) ^{\otimes N},
\label{75}
\end{equation}%
where $|j\rangle ,$ $j=1,...,d,$ are mutually orthogonal unit vectors in $%
\mathbb{C}^{d}.$ For each of these states, the trace norm of the source
operator (\ref{A20}) is upper bounded by%
\begin{align}
\left\Vert \widetilde{\tau }_{_{1\times S_{2}\times \cdots \times
S_{N}}}^{(\rho _{d})}\right\Vert _{1}& \leq 1+2^{N-1}(d-1),  \label{75'} \\
\left\Vert \widetilde{\tau }_{_{1\times S_{2}\times \cdots \times
S_{N}}}^{(\rho _{2}^{(gen)})}\right\Vert _{1}& \leq 1+2^{N-1}\left\vert \sin
\varphi \cos \varphi \right\vert ,  \notag
\end{align}%
and, in view of relations (\ref{75'}), (\ref{A9}), the analytical upper
bound (\ref{67}) implies 
\begin{align}
\mathrm{\Upsilon }_{\underset{N}{\underbrace{S\times \cdots \times S}}%
}^{(\rho _{d})}& \leq \min \{(2S-1)^{N-1},\text{ }1+2^{N-1}(d-1)\}
\label{78} \\
& \leq 1+2^{N-1}\left[ \min \{S^{N-1},d\}-1\right]  \notag
\end{align}%
and 
\begin{align}
\mathrm{\Upsilon }_{\underset{N}{\underbrace{S\times \cdots \times S}}%
}^{(\rho _{2}^{(gen)})}& \leq \min \left\{ (2S-1)^{N-1},\text{ }%
1+2^{N-1}\left\vert \sin 2\varphi \right\vert \right\}  \label{78_} \\
& \leq 1+2^{N-1}\text{ }\left\vert \sin 2\varphi \right\vert .  \notag
\end{align}%
The second lines in (\ref{78}), (\ref{78_}) are due to relation $%
(2S-1)^{N-1}\leq $ $2^{N-1}(S^{N-1}-1)+1$ that can be easily proved by
induction.

For an arbitrary $N$-partite state $\rho,$ the general analytical upper
bound (\ref{67}) implies the following new numerical upper estimate.

\begin{theorem}
For every state $\rho $ on a Hilbert space $\mathcal{H}_{1}\otimes \cdots
\otimes \mathcal{H}_{N}$ and arbitrary positive integers $%
S_{1},...,S_{N}\geq 1,$ the maximal $S_{1}\times \cdots \times S_{N}$%
-setting Bell violation satisfies relation%
\begin{align}
\mathrm{\Upsilon }_{S_{1}\times \cdots \times S_{N}}^{(\rho )}& \leq \min
\{\xi _{N},\theta _{N},\}  \label{79} \\
& \leq 1+2^{N-1}\left[ \min \left\{ \frac{d_{1}\cdot \ldots \cdot d_{N}}{%
\max_{n}d_{n}},\text{ \ }\frac{S_{1}\cdot \ldots \cdot S_{N}}{\max_{n}S_{n}}%
\right\} -1\right] ,  \notag
\end{align}%
where $d_{n}=\dim \mathcal{H}_{n},$ $n=1,...,N,$ and\ 
\begin{align}
\xi _{N}& =1+2^{N-1}\left( \frac{d_{1}\cdot \ldots \cdot d_{N}}{\max_{n}d_{n}%
}-1\right) ,  \label{80} \\
\theta _{N}& =(-1)^{N-1}+\min_{\substack{ \{n_{1}.\ldots ,n_{N-1}\}  \\ %
\subset \{1,...,N\}}}\sum_{k=0}^{N-2}(-1)^{k}\text{ }2^{N-1-k}\sum 
_{\substack{ n_{j_{1}}\neq \ldots \neq n_{j_{N-1-k}},  \\ n_{j}\in
\{n_{1},\ldots ,n_{N-1}\}}}S_{n_{j_{_{1}}}}\cdot \ldots \cdot
S_{n_{_{j_{_{N-1-k}}}}}.  \notag
\end{align}%
If, in particular, $d_{n}=d,$ $S_{n}=S,$ $\forall n,$ then%
\begin{align}
\mathrm{\Upsilon }_{\underset{N}{\underbrace{S\times \ldots \times S}}%
}^{(\rho )}& \leq \min \left\{ (2S-1)^{N-1},\text{ }2^{N-1}(d^{N-1}-1)+1%
\right\}  \label{81} \\
&  \notag \\
& \leq 1+2^{N-1}\left[ \left( \min \{S,d\}\right) ^{N-1}-1\right] .  \notag
\end{align}
\end{theorem}

\begin{proof}
From bound (\ref{67}) it follows%
\begin{align}
\mathrm{\Upsilon}_{S_{1}\times\cdots\times S_{N}}^{(\rho)} & \leq
\inf_{T_{_{S_{1}\times\cdots\times\underset{\overset{\uparrow}{n}}{1}%
\times\cdots\times S_{N}}}^{(\rho)},\forall n}\text{ }||T_{_{S_{1}\times
\cdots\times\underset{\overset{\uparrow}{n}}{1}\times\cdots\times
S_{N}}}^{(\rho)}\text{ }||_{cov}  \label{82} \\
& \leq\inf_{T_{_{S_{1}\times\cdots\times\underset{\overset{\uparrow}{n}}{1}%
\times\cdots\times S_{N}}}^{(\rho)},\forall n}\text{ }||T_{_{S_{1}\times%
\cdots\times\underset{\overset{\uparrow}{n}}{1}\times\cdots\times
S_{N}}}^{(\rho)}\text{ }||_{1}  \notag \\
& \leq\min\{||\tau_{_{S_{1}\times\cdots\times\underset{\overset{\uparrow}{n}}%
{1}\times\cdots\times S_{N}}}^{(\rho)}\text{ }||_{cov},\text{ }||\widetilde{%
\tau}_{_{S_{1}\times\cdots\times\underset{\overset{\uparrow}{n}}{1}%
\times\cdots\times S_{N}}}^{(\rho)}\text{ }||_{cov},\text{ \ }n=1,...,N\}, 
\notag
\end{align}
where $\tau_{_{S_{1}\times\cdots\times\underset{\overset{\uparrow}{n}}{1}%
\times\cdots\times S_{N}}}^{(\rho)}$ and $\widetilde{\tau}%
_{_{S_{1}\times\cdots\times\underset{\overset{\uparrow}{n}}{1}%
\times\cdots\times S_{N}}}^{(\rho)}$ are the specific source operators
constructed for an arbitrary $N$-partite state $\rho$ in appendix A. Taking
into the account the upper bounds (\ref{A9}), (\ref{A21}) for these source
operators and also relation $(2S-1)^{N-1}\leq2^{N-1}S^{N-1}-2^{N-1}+1,$ we
come to relation (\ref{79}), implying, in turn, (\ref{81}).
\end{proof}

Note that estimate (\ref{79}) implies that $\mathrm{\Upsilon}_{1\times
\cdots\times1\times S_{n}\times1\times\cdots\times1}^{(\rho)}=1$ for every $%
N $-partite state $\rho.$ In view of statement (b) of proposition 3, this
result agrees with statement (a) of proposition 5.

Theorem 4 implies.

\begin{corollary}
(a) For an arbitrary $N$-partite quantum state, violation of a Bell-type
inequality (either on correlation functions or on joint probabilities) for $%
S $ settings per site cannot exceed $(2S-1)^{N-1}$ even in case of an
infinite dimensional state and infinitely many outcomes. \newline
(b) For an arbitrary state $\rho$ on $(\mathbb{C}^{d})^{\otimes N}$,
violation of a Bell-type inequality (either on correlation functions or on
joint probabilities) is upper bounded by $2^{N-1}(d^{N-1}-1)+1$
independently on a number of settings and a number of outcomes at each site.
\end{corollary}

Let us now specify the general $N$-partite upper estimate (\ref{80}) for $%
N=2,3.$

\begin{corollary}
For every bipartite state $\rho$ and arbitrary positive integers $%
S_{1},S_{2}\geq1,$ 
\begin{equation}
\Upsilon_{S_{1}\times S_{2}}^{(\rho)}\leq2\min\left\{
S_{1},S_{2},d_{1},d_{2}\right\} -1.  \label{72}
\end{equation}
If, in particular, $d_{1}=d_{2}=d,$ $S_{1}=S_{2}=S,$ then 
\begin{equation}
\Upsilon_{S\times S}^{(\rho)}\leq2\min\{S,d\}-1.  \label{72'}
\end{equation}
\medskip For every tripartite state $\rho$ and arbitrary positive integers $%
S_{1},S_{2},S_{3}\geq1,$%
\begin{align}
\Upsilon_{S_{1}\times S_{2}\times S_{3}}^{(\rho)} & \leq\min\text{ }\left\{
\min_{_{\{n_{1},n_{2}\}\subset\{1,2,3%
\}}}(4S_{n_{1}}S_{n_{2}}-2(S_{n_{1}}+S_{n_{2}})+1),\text{ \ }4\frac{%
d_{1}d_{2}d_{3}}{\max_{n}d_{n}}-3\right\}  \label{85'} \\
& \leq4\min\left\{ \frac{S_{1}S_{2}S_{3}}{\max_{n}S_{n}},\text{ }\frac {%
d_{1}d_{2}d_{3}}{\max_{n}d_{n}}\right\} -3.  \notag
\end{align}
If, in particular, $d_{n}=d,$ $S_{n}=S,$ $\forall n,$ then 
\begin{align}
\Upsilon_{S\times S\times S}^{(\rho)} & \leq\min\left\{ (2S-1)^{2},\text{ }%
4d^{2}-3\right\}  \label{85} \\
& \leq4\left( \min\{S,d\}\right) ^{2}-3.  \notag
\end{align}
\end{corollary}

We stress that, in contrast to the bipartite and tripartite numerical
estimates found in Refs. \cite{junge, junge1, junge2, carlos} up to unknown
universal constants, our bipartite and tripartite numerical upper estimates (%
\ref{72}) - (\ref{85}) are exact.

\subsection{Discussion}

For bipartite and tripartite correlation scenarios with a finite number of
outcomes at each site, the numerical estimates on the maximal Bell
violations have been recently presented in Refs. \cite{junge, kaplan,
junge1, junge2, carlos}. The results of corollary 3 indicate.

\begin{itemize}
\item Our exact bipartite upper estimate (\ref{72}) improves the approximate
bipartite estimate $\preceq\min\{d,S\}$\ found in Ref. \cite{carlos}
(theorem 6.8) up to an unknown universal constant.

\item The bipartite upper estimates (in our notation)%
\begin{align}
\Upsilon_{S_{1}\times S_{2}}^{(\rho,\text{ }\Lambda)} & \leq2K_{G}+1,\text{
\ \ \ \ \ \ \ \ \ \ \ \ \ \ \ \ \ \ \ \ \ \ if \ }\left\vert \mathcal{A}%
\right\vert =\left\vert \mathcal{B}\right\vert =2, \\
\Upsilon_{S_{1}\times S_{2}}^{(\rho,\text{ }\Lambda)} & \leq2\left\vert 
\mathcal{A}\right\vert \left\vert \mathcal{B}\right\vert (K_{G}+1)-1,\text{
\ \ \ }\forall\left\vert \mathcal{A}\right\vert ,\left\vert \mathcal{B}%
\right\vert ,  \notag
\end{align}
derived in theorem 22 of Ref. \cite{kaplan} for arbitrary $S_{1},S_{2}\geq1$
and numbers $\left\vert \mathcal{A}\right\vert ,$ $\left\vert \mathcal{B}%
\right\vert $ of outcomes at Alice's and Bob's sites, are improved by our
bipartite upper estimate (\ref{72}) if%
\begin{align}
\min\{d_{1},d_{2}\} & <K_{G}+1, \\
\min\{d_{1},d_{2}\} & <\left\vert \mathcal{A}\right\vert \left\vert \mathcal{%
B}\right\vert (K_{G}+1),  \notag
\end{align}
respectively. Here, $K_{G}$ is the Grothendieck constant, and it is known 
\cite{Finch, krivine} that $K_{G}\in\lbrack1.676,1.783].$

\item From our exact tripartite upper estimates (\ref{85'}), (\ref{85}) it
follows that \emph{violation by a tripartite quantum state of a Bell-type
inequality for }$S$\emph{\ settings per site cannot exceed }$(2S-1)^{2}.$
Therefore, the tripartite lower estimate $\succeq\sqrt{d},$ found in theorem
1 of Ref. \cite{junge} for violation of some correlation Bell-type
inequality by some tripartite state on $\mathbb{C}^{d}\otimes\mathbb{C}%
^{D}\otimes \mathbb{C}^{D}$, is meaningful if only a number $S$ of settings
per site needed for such a violation in the corresponding Bell-type
inequality obeys relation%
\begin{equation}
(2S-1)^{2}\succeq\sqrt{d}.  \label{91}
\end{equation}
Thus, for an arbitrarily large tripartite violation argued in Ref. \cite%
{junge} to be reached, not only a Hilbert space dimension $d$ but also a
number $S$ of settings per site in the corresponding tripartite Bell-type
inequality must be large and the required growth of $S$ with respect to $d$
is given by (\ref{91}).
\end{itemize}

\section{Conclusions}

In the present paper, for the probabilistic description of a general
correlation scenario, we have introduced (\emph{definition 5}) a new
simulation model, \emph{a} \emph{local quasi hidden variable (LqHV) model, }%
where locality and the measure-theoretic structure inherent to an LHV model
are preserved but positivity of a simulation measure is dropped.

We have specified (\emph{theorem 1}) a necessary and sufficient condition
for LqHV modelling and, based on this, proved (\emph{theorem 2}) that every
quantum correlation scenario admits an LqHV simulation.

Via the LqHV approach, we have constructed analogs (Eq. (\ref{52'})) of
Bell-type inequalities for an $N$-partite quantum state and found (\emph{%
theorem 3}) a new analytical upper bound on the maximal violation by an $N$%
-partite state of all $S_{1}\times\cdots\times S_{N}$-setting Bell-type
inequalities -- either on correlation functions or on joint probabilities
and for outcomes of an arbitrary spectral type, discrete or continuous.

This analytical upper bound is based on the new state dilation
characteristics (\emph{definitions 1-4}) introduced in the present paper and
this allows us:

\begin{itemize}
\item to trace (\emph{propositions 5, 6}) $N$-partite quantum states
admitting an $S_{1}\times\cdots\times S_{N}$- setting LHV description;

\item to find the exact numerical upper estimates (Eqs. (\ref{sing}), (\ref%
{78}), (\ref{78_})) on the maximal $S_{1}\times\cdots\times S_{N}$-setting
Bell violations for some concrete $N$-partite quantum states used in quantum
information processing;

\item to find (\emph{theorem 4}) the exact numerical upper estimate on the
maximal $S_{1}\times\cdots\times S_{N}$- setting Bell violation for an
arbitrary $N$-partite quantum state, in particular, to show (\emph{corollary
2}) that violation by an $N$-partite quantum state of a Bell-type inequality
(either on correlation functions or on joint probabilities) for $S$ settings
per site is upper bounded by $(2S-1)^{N-1}$ even in case of an infinite
dimensional quantum state and infinitely many outcomes.
\end{itemize}

Specified (\emph{corollary 3}) for $N=2,3$, our exact $N$-partite numerical
upper estimate (Eq. (\ref{79})) improves the bipartite numerical upper
estimates in Refs. \cite{kaplan, carlos} and clarifies the range of
applicability of the approximate tripartite numerical lower estimate in Ref. 
\cite{junge}.

\begin{acknowledgement}
I am grateful to the organizers of the Operator Structures in Quantum
Information Workshop, held at the Fields Institute in Toronto on July 6 -10,
2009, where part of this paper covering a bipartite case was presented. I am
also thankful to Professor A. M. Chebotarev for the useful discussion.
\end{acknowledgement}

\appendix

\section{proofs for section 2}

In this appendix, we prove proposition 1 on the existence of source
operators for an $N$-partite quantum state and introduce some source
operators which are different from those constructed in the proof of
proposition 1 and are needed for our consideration in section 2 of appendix
B. We also prove lemma1 on the properties of the covering norm.

\subsection{Proof of proposition 1}

For a state $\rho$ on a Hilbert space $\mathcal{H}_{1}\otimes\cdots \otimes%
\mathcal{H}_{N},$ consider its decomposition 
\begin{align}
\rho & =\sum\eta_{mm_{1},...,kk_{1}}|e_{m}^{(1)}\rangle\langle
e_{m_{1}}^{(1)}|\otimes\cdots\otimes|e_{k}^{(N)}\rangle\langle
e_{k_{1}}^{(N)}|,  \tag{A1}  \label{A1} \\
\eta_{mm_{1},...,kk_{1}} &
=\sum_{i}\alpha_{i}\varsigma_{m...k}^{(i)}\varsigma_{m_{1}...k_{1}}^{%
\ast(i)},\text{ \ \ \ }\alpha_{i}>0,\text{ \ }\sum_{i}\alpha_{i}=1,\text{\ }%
\sum_{m,...,k}\left\vert \varsigma _{m...k}^{(i)}\right\vert ^{2}=1,  \notag
\end{align}
in orthonormal bases $\{e_{m}^{(n)}\in\mathcal{H}_{n}\},$ $n=1,...,N.$

Let $N=2.$ For a bipartite state $\rho,$ denote by $\rho_{n}$ the reduced
state on $\mathcal{H}_{n},$ $n=1,2,$ and introduce on a Hilbert space $%
\mathcal{H}_{1}^{\otimes S_{1}}\otimes\mathcal{H}_{2}^{\otimes S_{2}}$ the
self-adjoint operator 
\begin{align}
\tau_{_{S_{1}\times S_{2}}}^{(\rho)} & =\sum\eta_{mm_{1},kk_{1}}\left[
|e_{m}^{(1)}\rangle\langle e_{m_{1}}^{(1)}|\otimes\mathbb{\sigma}%
_{1}^{\otimes(S_{1}-1)}\right] _{\mathrm{sym}}\otimes\left[
|e_{k}^{(2)}\rangle\langle e_{k_{1}}^{(2)}|\otimes\mathbb{\sigma}%
_{2}^{\otimes(S_{2}-1)}]\right] _{\mathrm{sym}}  \tag{A2}  \label{A2} \\
& -(S_{2}-1)\rho_{1}^{\otimes S_{1}}\otimes\mathbb{\sigma}_{2}^{\otimes
S_{2}}-(S_{1}-1)\mathbb{\sigma}_{1}^{\otimes S_{1}}\otimes\rho_{2}^{\otimes
S_{2}}  \notag \\
& -(S_{1}-1)(S_{2}-1)\mathbb{\sigma}_{1}^{\otimes S_{1}}\otimes \mathbb{%
\sigma}_{2}^{\otimes S_{2}},  \notag
\end{align}
where $\sigma_{n}\ $is a state on $\mathcal{H}_{n}$ and notation $%
[\cdot]_{sym}$ means symmetrization on $\mathcal{H}_{n}^{\otimes S_{n}}$.
For example, $[X_{1}\otimes X_{2}]_{sym}$ $:=X_{1}\otimes X_{2}+X_{2}\otimes
X_{1}.$ It is easy to verify that (\ref{A2}) represents an $S_{1}\times
S_{2} $-setting source operator for state (\ref{A1}) specified with $N=2$
and the trace norm of this source operator satisfies relation 
\begin{equation}
1\leq\left\Vert \tau_{_{S_{1}\times S_{2}}}^{(\rho)}\right\Vert _{1}\leq
2S_{1}S_{2}-1.  \tag{A3}  \label{A3}
\end{equation}

Let $N=3.$ For a tripartite state $\rho,$ in addition to the above notation $%
\rho_{n}$ for the reduced state on $\mathcal{H}_{n}$, denote by $\rho_{1,n}$
the reduced state on $\mathcal{H}_{1}\otimes\mathcal{H}_{n}.$ For short of
notation, we further take one of settings to be equal only to $1$, say $%
S_{1}=1.$ Introduce on a Hilbert space $\mathcal{H}_{1}\otimes\mathcal{H}%
_{2}^{\otimes S_{2}}\otimes\mathcal{H}_{3}^{\otimes S_{3}}$ the self-adjoint
operator%
\begin{align}
\tau_{_{1\times S_{2}\times S_{3}}}^{(\rho)} &
=\sum\eta_{mm_{1},jj_{1},kk_{1}}|e_{m}^{(1)}\rangle\langle
e_{m_{1}}^{(1)}|\otimes\left[ |e_{j}^{(2)}\rangle\langle
e_{j_{1}}^{(2)}|\otimes\mathbb{\sigma}_{2}^{\otimes(S_{2}-1)}\right] _{%
\mathrm{sym}}  \tag{A4}  \label{A4} \\
& \otimes\left[ |e_{k}^{(3)}\rangle\langle e_{k_{1}}^{(3)}|\otimes \mathbb{%
\sigma}_{3}^{\otimes(S_{3}-1)}\right] _{\mathrm{sym}}-(S_{3}-1)\tau_{_{1%
\times S_{2}}}^{(\rho_{1,2})}\otimes\mathbb{\sigma}_{3}^{\otimes S_{3}} 
\notag \\
& -(S_{2}-1)\overset{\tau_{_{1\times S_{3}}}^{(\rho_{1,3})}}{\overbrace {%
\checkmark\otimes\mathbb{\sigma}_{2}^{\otimes S_{2}}\otimes\checkmark}}%
-(S_{2}-1)(S_{3}-1)\rho_{1}^{\otimes S_{1}}\otimes\mathbb{\sigma}%
_{2}^{\otimes S_{2}}\otimes\mathbb{\sigma}_{3}^{\otimes S_{3}},  \notag
\end{align}
where: (i) $\tau_{_{1\times S_{n}}}^{(\rho_{1,n})}$ is the $1\times S_{n}$%
-source operator (\ref{A2}) specified for the reduced state $\rho_{1,n}$ on $%
\mathcal{H}_{1}\otimes\mathcal{H}_{n}$; (ii) notation in the third line
means operator derived by insertion of term $\otimes\mathbb{\sigma}%
_{2}^{\otimes S_{2}}\otimes$ into each term of the tensor product
decomposition of the source operator $\tau_{_{1\times S_{3}}}^{(\rho_{1,3})}$
for state $\rho _{1,3}$ on $\mathcal{H}_{1}\otimes\mathcal{H}_{3}$. It is
easy to verify that (\ref{A4}) represents an $1\times S_{2}\times S_{3}$%
-setting source operator for state (\ref{A1}) specified with $N=3$.
Substituting (\ref{A2}) into (\ref{A4}), taking into the account (\ref{4})
and evaluating the negative part of the self-adjoint operator (\ref{A4}), we
derive%
\begin{align}
\left\Vert \tau_{_{1\times S_{2}\times S_{3}}}^{(\rho)}\right\Vert _{1} &
\leq1+2(S_{2}-1)S_{3}+2(S_{3}-1)S_{2}  \tag{A5}  \label{A5} \\
& =4S_{2}S_{3}-2(S_{2}+S_{3})+1  \notag \\
& \leq4S_{2}S_{3}-3.  \notag
\end{align}

Let $N=4.$ For a quadripartite state $\rho,$ denote by $\rho_{n},$ $%
\rho_{n_{1},n_{2}},$ $\rho_{n_{1},n_{2},n_{3}}$ the reduced states on $%
\mathcal{H}_{n},$ $\mathcal{H}_{n_{1}}\otimes\mathcal{H}_{n_{2}}$ and $%
\mathcal{H}_{n_{1}}\otimes\mathcal{H}_{n_{2}}\otimes\mathcal{H}_{n_{3}},$
respectively and consider on $\mathcal{H}_{1}\otimes\mathcal{H}_{2}^{\otimes
S_{2}}\otimes\mathcal{H}_{3}^{\otimes S_{3}}\otimes\mathcal{H}_{4}^{\otimes
S_{4}}$ the self-adjoint operator 
\begin{align}
\tau_{_{1\times S_{2}\times S_{3}\times S_{4}}}^{(\rho)} & =\sum\eta
_{mm_{1},jj_{1},l,l_{1}kk_{1}}|e_{m}^{(1)}\rangle\langle
e_{m_{1}}^{(1)}|\otimes\left[ |e_{j}^{(2)}\rangle\langle
e_{j_{1}}^{(2)}|\otimes\mathbb{\sigma}_{2}^{\otimes(S_{2}-1)}]\right] _{%
\mathrm{sym}}  \notag \\
& \otimes\left[ \lbrack e_{l}^{(3)}\rangle\langle e_{l_{1}}^{(3)}|\otimes%
\mathbb{\sigma}_{3}^{\otimes(S_{3}-1)}\right] _{\mathrm{sym}}\otimes\left[
|e_{k}^{(4)}\rangle\langle e_{k_{1}}^{(4)}|\otimes \mathbb{\sigma}%
_{4}^{\otimes(S_{4}-1)}\right] _{\mathrm{sym}}\text{ }  \notag \\
& -\text{ }(S_{4}-1)\tau_{_{1\times S_{2}\times
S_{3}}}^{(\rho_{1,2,3})}\otimes\mathbb{\sigma}_{4}^{\otimes S_{4}}-(S_{2}-1)%
\overset{\tau_{_{1\times S_{3}\times S_{4}}}^{(\rho_{1,3,4})}}{\overbrace{%
\checkmark\otimes \mathbb{\sigma}_{2}^{\otimes
S_{2}}\otimes\checkmark\otimes\checkmark}}  \notag \\
& -(S_{3}-1)\overset{\tau_{_{1\times S_{2}\times S_{4}}}^{(\rho_{1,2,4})}}{%
\overbrace{\checkmark\otimes\checkmark\otimes\mathbb{\sigma}_{3}^{\otimes
S_{3}}\otimes\checkmark}}  \notag \\
& -(S_{3}-1)(S_{4}-1)\tau_{_{1\times S_{2}}}^{(\rho_{1,2})}\otimes \mathbb{%
\sigma}_{3}^{\otimes S_{3}}\otimes\mathbb{\sigma}_{4}^{\otimes S_{4}} 
\tag{A6}  \label{A6} \\
& -(S_{2}-1)(S_{3}-1)\overset{\tau_{_{1\times S_{4}}}^{(\rho_{1,4})}}{%
\overbrace{\checkmark\otimes\mathbb{\sigma}_{2}^{\otimes S_{2}}\otimes%
\mathbb{\sigma}_{3}^{\otimes S_{3}}\otimes\checkmark}}  \notag \\
& -(S_{2}-1)(S_{4}-1)\overset{\tau_{_{1\times S_{3}}}^{(\rho_{1,3})}}{%
\overbrace{\checkmark\otimes\mathbb{\sigma}_{2}^{\otimes
S_{2}}\otimes\checkmark}\otimes\mathbb{\sigma}_{4}^{\otimes S_{4}}}  \notag
\\
& -(S_{2}-1)(S_{3}-1)(S_{4}-1)\rho_{1}^{\otimes S_{1}}\otimes\mathbb{\sigma }%
_{2}^{\otimes S_{2}}\otimes\mathbb{\sigma}_{3}^{\otimes S_{3}}\otimes 
\mathbb{\sigma}_{4}^{\otimes S_{4}},  \notag
\end{align}
where $\tau_{_{1\times S_{n}}}^{(\rho_{1,n})}$ is the $1\times S_{n}$%
-setting source operator (\ref{A2}) specified for the reduced state $%
\rho_{1,n}$ and $\tau_{_{1\times S_{n_{1}}\times
S_{n_{2}}}}^{(\rho_{1,n_{1},n_{2}})}$ is the source operator (\ref{A5}) for
the reduced state $\rho_{1,n_{1},n_{2}}$. It is easy to verify that (\ref{A6}%
) represents an $1\times S_{2}\times S_{3}\times S_{4}$-setting source
operator for state (\ref{A1}) specified with $N=4$. Substituting (\ref{A2}),
(\ref{A4}) into (\ref{A6}) and evaluating the negative part of the
self-adjoint operator (\ref{A6}), we derive%
\begin{align}
\left\Vert \tau_{_{1\times S_{2}\times S_{3}\times
S_{4}}}^{(\rho)}\right\Vert _{1} &
\leq1+2(S_{2}-1)S_{3}S_{4}+2(S_{3}-1)S_{2}S_{4}  \tag{A7}  \label{A7} \\
& +2(S_{4}-1)S_{2}S_{3}+2(S_{2}-1)(S_{3}-1)(S_{4}-1)  \notag \\
& =8S_{2}S_{3}S_{4}-4(S_{2}S_{3}+S_{2}S_{4}+S_{3}S_{4})  \notag \\
& +2(S_{2}+S_{3}+S_{4})-1.  \notag
\end{align}

Thus, in view of (\ref{A3}), (\ref{A5}), (\ref{A7}),%
\begin{align}
\left\Vert \tau _{_{1\times S_{2}}}^{(\rho )}\right\Vert _{1}& \leq 2S_{2}-1,
\tag{A8}  \label{A8} \\
\left\Vert \tau _{_{1\times S_{2}\times S_{4}}}^{(\rho )}\right\Vert _{1}&
\leq 4S_{2}S_{3}-2(S_{2}+S_{3})+1\leq 4S_{2}S_{3}-3,  \notag \\
\left\Vert \tau _{_{1\times S_{2}\times S_{3}\times S_{4}}}^{(\rho
)}\right\Vert _{1}& \leq 8S_{2}S_{3}S_{4}-4(S_{2}S_{3}+S_{2}S_{4}+S_{3}S_{4})
\notag \\
& +2(S_{2}+S_{3}+S_{4})-1\leq 8S_{2}S_{3}S_{4}-7.  \notag
\end{align}%
All these bounds are tight in the sense they imply $\left\Vert \tau
_{_{1\times \cdots \times 1}}^{(\rho )}\right\Vert _{1}=1.$ The
generalization to an $N$-partite case with setting $S_{1}=1$ is
straightforward and gives%
\begin{align}
\left\Vert \tau _{_{1\times S_{2}\times \cdots \times S_{N}}}^{(\rho
)}\right\Vert _{1}& \leq \sum_{k=0}^{N-2}(-1)^{k}\text{ }2^{N-1-k}\sum 
_{\substack{ n_{1}\neq \ldots \neq n_{N-1-k},  \\ n_{j}=2,...,N}}%
S_{n_{1}}\cdot \ldots \cdot S_{n_{N-1-k}}  \tag{A9}  \label{A9} \\
& +(-1)^{N-1},  \notag \\
\left\Vert \tau _{_{1\times S\times \cdots \times S}}^{(\rho )}\right\Vert
_{1}& \leq (2S-1)^{N-1}.  \notag
\end{align}

The constructed source operators (\ref{A2}), (\ref{A4}), (\ref{A6}), (\ref%
{A9}) prove the statement of proposition 1.

\subsubsection{Other\ examples of source operators\ }

For our consideration in section 2 of appendix B, let us also construct
source operators of a type $\widetilde{\tau}$ different from type $\tau$ in
Eqs. (\ref{A2}), (\ref{A4}), (\ref{A6}), (\ref{A9}). Denote $d_{n}:=\dim 
\mathcal{H}_{n}$, $n=1,...,N,$ and assume that $\max_{n}d_{n}=d_{1}$.

Let $N=2$ and $\psi\in\mathcal{H}_{1}\otimes\mathcal{H}_{2}$. For a pure
bipartite state $|\psi\rangle\langle\psi|$, consider its Schmidt
decomposition 
\begin{equation}
|\psi\rangle\langle\psi|=\sum\xi_{j}\xi_{j_{1}}|g_{j}^{(1)}\rangle\langle
g_{j_{1}}^{(1)}|\otimes|g_{j}^{(2)}\rangle\langle g_{j_{1}}^{(2)}|,\ \ \ \
\xi_{j}>0,\ \ \ \ \sum\xi_{j}^{2}=1,  \tag{A10}  \label{A10}
\end{equation}
where the sum is taken over $j,j_{1}=1,...,d_{2}$ and $\{g_{j}^{(n)}\}$ is
an orthonormal base in $\mathcal{H}_{n},$ $n=1,2.$ \ Introduce the
self-adjoint operator%
\begin{equation}
\widetilde{\tau}_{_{1\times S_{2}}}^{|\psi\rangle\langle\psi|}=\sum_{j}\xi
_{j}^{2}\text{ }|g_{j}^{(1)}\rangle\langle
g_{j}^{(1)}|\otimes(|g_{j}^{(2)}\rangle\langle g_{j}^{(2)}|)^{\otimes
S_{2}}+\sum_{j\neq j_{1}}\xi _{j}\xi_{j_{1}}|g_{j}^{(1)}\rangle\langle
g_{j_{1}}^{(1)}|\otimes W_{jj_{1}}^{(2,S_{2})},  \tag{A11}  \label{A11}
\end{equation}
where 
\begin{align}
2W_{j\neq j_{1}}^{(2,S_{2})} & =\frac{\left(
|g_{j}^{(2)}+g_{j_{1}}^{(2)}\rangle\langle
g_{j}^{(2)}+g_{j_{1}}^{(2)}|\right) ^{\otimes S_{2}}}{2^{S_{2}}}-\frac{%
\left( |g_{j}^{(2)}-g_{j_{1}}^{(2)}\rangle\langle
g_{j}^{(2)}-g_{j_{1}}^{(2)}|\right) ^{\otimes S_{2}}}{2^{S_{2}}}  \tag{A12}
\label{A12} \\
& -i\frac{\left( |g_{j}^{(2)}+ig_{j_{1}}^{(2)}\rangle\langle
g_{j}^{(2)}+ig_{j_{1}}^{(2)}|\right) ^{\otimes S_{2}}}{2^{S_{2}}}+i\frac{%
\left( |g_{j}^{(2)}-ig_{j_{1}}^{(2)}\rangle\langle
g_{j}^{(2)}-ig_{j_{1}}^{(2)}|\right) ^{\otimes S_{2}}}{2^{S_{2}}}\text{ }, 
\notag \\
\left( W_{jj_{1}}^{(2,S_{2})}\right) ^{\ast} & =W_{j_{1}j}^{(2,S_{2})}. 
\notag
\end{align}
It is easy to verify that (\ref{A11}) represents an $1\times S_{2}$-setting
source operator for state (\ref{A10}) and 
\begin{align}
\left\Vert \widetilde{\tau}_{_{1\times S_{2}}}^{|\psi\rangle\langle\psi
|}\right\Vert _{1} & \leq1+2\sum_{j\neq j_{1}}\xi_{j}\xi_{j_{1}}=2\left(
\sum\xi_{j}\right) ^{2}-1  \tag{A13}  \label{A13} \\
& \leq2d_{2}-1  \notag
\end{align}
for any $S_{2}\geq1.$ By convexity, for an arbitrary bipartite state $%
\rho=\sum_{i}\alpha_{i}|\psi_{i}\rangle\langle\psi_{i}|,$ operator $%
\widetilde{\tau}_{_{1\times S_{2}}}^{(\rho)}=\sum\alpha_{i}\widetilde{\tau }%
_{_{1\times S_{2}}}^{(|\psi_{i}\rangle\langle\psi_{i}|)}$ represents an $%
1\times S_{2}$-setting source operator and 
\begin{equation}
\left\Vert |\widetilde{\tau}_{_{1\times S_{2}}}^{(\rho)}\right\Vert
_{1}\leq2d_{2}-1,\text{ \ \ }\forall S_{2}\geq1.  \tag{A14}  \label{A!4}
\end{equation}

Let $N=3$. For state (\ref{A1}) with $N=3$ introduce the self-adjoint
operator 
\begin{equation}
\widetilde{\tau}_{_{1\times S_{2}\times S_{3}}}^{(\rho)}=\sum\eta
_{mm_{1},jj_{1}...kk_{1}}|e_{m}^{(1)}\rangle\langle e_{m_{1}}^{(1)}|\otimes
W_{jj_{1}}^{(2,S_{2})}\otimes W_{kk_{1}}^{(3,S_{3})},  \tag{A15}  \label{A15}
\end{equation}
where $W_{ll}^{(n,S_{n})}:=\left( |e_{l}^{(n)}\rangle\langle
e_{l}^{(n)}|\right) ^{\otimes S_{n}}$ and operator $W_{l\neq
l_{1}}^{(n,S_{n})}$ is defined by (\ref{A12}) via replacements $2\rightarrow
n,$ $S_{2}\rightarrow S_{n},$ $g_{j}^{(2)}\rightarrow e_{j}^{(n)}.$ It is
easy to verify that (\ref{A15}) is an $1\times S_{2}\times S_{3}$-setting
source operator for state (\ref{A1}) with $N=3.$ Splitting (\ref{A15}) into
four sums 
\begin{align}
\widetilde{\tau}_{_{1\times S_{2}\times S_{3}}}^{(\rho)} & =\sum
_{i,j,k}\alpha_{i}(\beta_{jk}^{(i)})^{2}|\phi_{jk}^{(i)}\rangle\langle
\phi_{jk}^{(i)}|\otimes(|e_{j}^{(2)}\rangle\langle e_{j}^{(2)}|)^{\otimes
S_{2}}\otimes(|e_{k}^{(3)}\rangle\langle e_{k}^{(3)}|)^{\otimes S_{3}} 
\tag{A16}  \label{A16} \\
& +\sum_{i,j\neq
j_{1},k}\alpha_{i}\beta_{jk}^{(i)}\beta_{j_{1}k}^{(i)}|\phi_{jk}^{(i)}%
\rangle\langle\phi_{j_{1}k}^{(i)}|\otimes
W_{jj_{1}}^{(2,S_{2})}\otimes(|e_{k}^{(3)}\rangle\langle
e_{k}^{(3)}|)^{\otimes S_{3}}  \notag \\
& +\sum_{i,j,k\neq
k_{1}}\alpha_{i}\beta_{jk}^{(i)}\beta_{jk_{1}}^{(i)}|\phi_{jk}^{(i)}\rangle%
\langle\phi_{jk_{1}}^{(i)}|\otimes(|e_{j}^{(2)}\rangle\langle
e_{j}^{(2)}|)^{\otimes S_{1}}\otimes W_{kk_{1}}^{(3,S_{3})}  \notag \\
& +\sum_{i,j\neq j_{1},k\neq k_{1}}\alpha_{i}\beta_{jk}^{(i)}\beta
_{j_{1}k_{1}}^{(i)}|\phi_{jk}^{(i)}\rangle\langle\phi_{j_{1}k_{1}}^{(i)}|%
\otimes W_{jj_{1}}^{(2,S_{2})}\otimes W_{kk_{1}}^{(3,S_{3})},  \notag
\end{align}
where, for any index $i,$%
\begin{align}
\phi_{jk}^{(i)} & =\frac{1}{\beta_{jk}^{(i)}}\sum_{m}%
\varsigma_{mjk}^{(i)}e_{m}^{(1)},\ \ \ ||\phi_{jk}^{(i)}||\text{ }=1, 
\tag{A17}  \label{A17} \\
\beta_{jk}^{(i)} & =\left( \dsum
\limits_{m}|\varsigma_{mjk}^{(i)}|^{2}\right) ^{1/2},\text{ \ \ }%
\sum_{j,k}(\beta _{jk}^{(i)})^{2}=1,  \notag
\end{align}
and taking into the account that%
\begin{align}
\left\Vert |\phi_{jk}^{(i)}\rangle\langle\phi_{j_{1}k_{1}}^{(i)}|\right\Vert
_{1} & =1,\text{ \ \ }\sum\beta_{jk}^{(i)}\leq\sqrt{d_{2}d_{3}},  \tag{A18}
\label{A18} \\
\left\Vert W_{ll}^{(n,S_{n})}\right\Vert _{1} & =1,\text{ \ }\left\Vert
W_{l\neq l_{1}}^{(n,S_{n})}\right\Vert _{1}\leq2\text{, \ \ }n=2,3,  \notag
\end{align}
we derive%
\begin{equation}
\left\Vert \widetilde{\tau}_{_{1\times S_{2}\times
S_{3}}}^{(\rho)}\right\Vert _{1}\leq4d_{2}d_{3}-3  \tag{A19}  \label{A19}
\end{equation}
for any $S_{2},S_{3}\geq1.$

The generalization of (\ref{A15}), (\ref{A19}) for $N\geq4$ is
straightforward and gives the source operator%
\begin{equation}
\widetilde{\tau}_{_{1\times S_{2}\times\cdots\times S_{N}}}^{(\rho)}=\sum
\eta_{mm_{1},jj_{1,\ldots},kk_{1}}|e_{m}^{(1)}\rangle\langle
e_{m_{1}}^{(1)}|\otimes W_{jj_{1}}^{(2,S_{2})}\otimes\cdots\otimes
W_{kk_{1}}^{(N,S_{N})},  \tag{A20}  \label{A20}
\end{equation}
with the trace norm 
\begin{equation}
\left\Vert \widetilde{\tau}_{_{1\times S_{2}\times\cdots\times
S_{N}}}^{(\rho)}\right\Vert _{1}\leq2^{N-1}(d_{2}\cdot\ldots\cdot d_{N}-1)+1.
\tag{A21}  \label{A21}
\end{equation}

\subsubsection{\textbf{Source operator for the singlet}}

For the two-qubit singlet $\psi_{singlet}=\frac{1}{\sqrt{2}}(e_{1}\otimes
e_{2}-e_{2}\otimes e_{1})$, the self-adjoint operator%
\begin{align}
T_{1\times2}^{(\rho_{singlet})} & =\frac{1}{2}|e_{1}\rangle\langle
e_{1}|\otimes|e_{2}\rangle\langle e_{2}|\otimes|e_{2}\rangle\langle e_{2}|+%
\frac{1}{2}|e_{2}\rangle\langle e_{2}|\otimes|e_{1}\rangle\langle
e_{1}|\otimes|e_{1}\rangle\langle e_{1}|  \tag{A22}  \label{A22} \\
& -\frac{1}{2}|e_{1}\rangle\langle e_{2}|\otimes|e_{2}\rangle\langle
e_{1}|\otimes\frac{\mathbb{I}_{\mathbb{C}^{2}}}{2}-\frac{1}{2}%
|e_{2}\rangle\langle e_{1}|\otimes|e_{1}\rangle\langle e_{2}|\otimes\frac {%
\mathbb{I}_{\mathbb{C}^{2}}}{2}  \notag \\
& -\frac{1}{2}|e_{1}\rangle\langle e_{2}|\otimes\frac{\mathbb{I}_{\mathbb{C}%
^{2}}}{2}\otimes|e_{2}\rangle\langle e_{1}|\text{ }-\frac{1}{2}%
|e_{2}\rangle\langle e_{1}|\otimes\frac{\mathbb{I}_{\mathbb{C}^{2}}}{2}%
\otimes|e_{1}\rangle\langle e_{2}|  \notag
\end{align}
on $(\mathbb{C}^{2})^{\otimes3}$ represents an $1\times2$-setting source
operator. Calculating the eigenvalues of this source operator, we derive $%
\lambda_{1,2,3,4}=0,$ $\lambda_{5,6}=\frac{1-\sqrt{3}}{4}$ and $\lambda
_{7,8}=\frac{1+\sqrt{3}}{4}.$ Hence,%
\begin{equation}
\left\Vert T_{1\times2}^{(\rho_{singlet})}\right\Vert _{1}=\sqrt{3}. 
\tag{A23}  \label{A23}
\end{equation}

\subsection{Proof of lemma 1}

\emph{Property 1.}\textbf{\ }The first left-hand side inequality of property
(1) is trivial. The last right-hand side inequality is due to the fact that,
for each $W\in\mathcal{T}_{\mathcal{G}_{1}\otimes\mathcal{\cdots}\otimes%
\mathcal{G}_{m}}^{(sa)},$ operator $|W|$ is one of its coverings, so that,
by (\ref{15}) and definition of the trace norm, we have 
\begin{equation}
\left\Vert W\right\Vert _{cov}\leq\mathrm{tr}[|W|]:=\left\Vert W\right\Vert
_{1}.  \tag{A24}  \label{A24}
\end{equation}
In order to prove 
\begin{equation}
\sup_{X_{j}=X_{j}^{\ast},\text{ }||X_{j}||=1,\forall j}\text{ }\left\vert 
\text{ }\mathrm{tr}[W\{X_{1}\otimes\cdots\otimes X_{m}\}]\right\vert
\leq\left\Vert W\right\Vert _{cov},  \tag{A25}  \label{A25}
\end{equation}
where supremum is taken over all self-adjoint bounded linear operators $%
X_{j},$ $j=1,...,N,$ with the operator norm $\left\Vert X_{j}\right\Vert =1,$
let us represent operator $W$ $\in\mathcal{T}_{\mathcal{G}_{1}\otimes 
\mathcal{\cdots}\otimes\mathcal{G}_{m}}^{(sa)}$ via decomposition (\ref{11})
with an arbitrary a trace class covering $W_{cov}$. We have 
\begin{align}
\sup\left\vert \text{ }\mathrm{tr}[W\{X_{1}\otimes\cdots\otimes
X_{m}\}]\right\vert \text{ } & \leq\frac{1}{2}\sup\left\vert \text{ }\mathrm{%
tr}[(W_{cov}+W)\{X_{1}\otimes\cdots\otimes X_{m}\}]\right\vert \text{ } 
\tag{A26}  \label{A26} \\
& +\frac{1}{2}\sup\left\vert \text{ }\mathrm{tr}[(W_{cov}-W)\{X_{1}\otimes%
\cdots\otimes X_{m}\}]\right\vert .  \notag
\end{align}
Applying to each term in (\ref{A26}) the spectral theorem 
\begin{equation}
X_{k}=\dint \lambda\mathbb{E}_{X_{k}}(\mathrm{d}\lambda),  \tag{A27}
\label{A27}
\end{equation}
where $\mathbb{E}_{X_{k}}$ is the spectral (projection-valued) measure on $(%
\mathbb{R},\mathcal{B}_{\mathbb{R}})$ for each self-adjoint bounded linear
operator $X_{k},$ \ $k=1,...,m,$ and taking into the account that $%
W_{cov}\pm W\overset{\otimes}{\geq}0$ and spectrum $\sigma(X_{k})\subseteq%
\lbrack-1,1]$ for each $k,$ we derive%
\begin{align}
& \sup\left\vert \text{ }\mathrm{tr}[(W_{cov}\pm W)\{X_{1}\otimes
\cdots\otimes X_{m}\}]\right\vert  \tag{A28}  \label{A28} \\
& \leq\sup\text{ }\left\vert \dint \lambda_{1}\cdot\ldots\cdot\lambda_{m}%
\text{ }\mathrm{tr}\left[ \left( W_{cov}\pm W\right) \{\mathbb{E}_{X_{1}}(%
\mathrm{d}\lambda_{1})\otimes \cdots\otimes\mathbb{E}_{X_{m}}(\mathrm{d}%
\lambda_{m})\}\right] \text{ }\right\vert  \notag \\
& \leq\mathrm{tr}\left[ W_{cov}\pm W\right] .  \notag
\end{align}
Substituting (\ref{A28}) into (\ref{A26}), we have 
\begin{equation}
\sup_{X_{j}=X_{j}^{\ast},\text{ }|\left\Vert X_{j}\right\Vert =1,\forall j}%
\text{ }\left\vert \mathrm{tr}[W\left\{ X_{1}\otimes\cdots\otimes
X_{m}\right\} ]\text{ }\right\vert \leq\mathrm{tr}\left[ W_{cov}\right] 
\tag{A29}  \label{A29}
\end{equation}
for each trace class covering $W_{cov}$ of an operator $W\in\mathcal{T}_{%
\mathcal{G}_{1}\otimes\mathcal{\cdots}\otimes\mathcal{G}_{m}}^{(sa)}.$ This
relation and definition (\ref{15}) of the covering norm imply inequality (%
\ref{A25}).

\emph{Property 2.} If a self-adjoint trace class operator $W$ is tensor
positive, then $\mathrm{tr}[W]\geq0$ and $W$ is itself one of its possible
coverings. Therefore, by (\ref{15}), $\left\Vert W\right\Vert _{cov}\leq%
\mathrm{tr}[W].$ This and property (1) of lemma 1 prove property (2).

\emph{Property 3.} Let $W_{red}\in\mathcal{T}_{\mathcal{G}_{k_{1}}\otimes%
\mathcal{\cdots}\otimes\mathcal{G}_{k_{j}}}^{(sa)},$ $1\leq
k_{1}<\cdots<k_{j}\leq m,$ $j<m,$ be the self-adjoint trace class operator
reduced from a self-adjoint trace class operator $W$ on a Hilbert space $%
\mathcal{G}_{1}\otimes\mathcal{\cdots}\otimes\mathcal{G}_{m}.$ The left-hand
side inequality in property (3) follows from property (1) and relation $%
\mathrm{tr}[W_{red}]=\mathrm{tr}[W]$.

Further, if $(W_{cov})_{red}$ $\in\mathcal{T}_{\mathcal{G}_{k_{1}}\otimes%
\mathcal{\cdots}\otimes\mathcal{G}_{k_{j}}}$ is the operator reduced from a
trace class covering $W_{cov}$ of operator $W\in\mathcal{T}_{\mathcal{G}%
_{1}\otimes\mathcal{\cdots}\otimes\mathcal{G}_{m}}^{(sa)},$ then $%
(W_{cov})_{red}$ is one of trace class coverings of operator $W_{red}\in%
\mathcal{T}_{\mathcal{G}_{k_{1}}\otimes\mathcal{\cdots}\otimes \mathcal{G}%
_{k_{j}}}^{(sa)}$. Therefore, for each $W\in\mathcal{T}_{\mathcal{G}%
_{1}\otimes\mathcal{\cdots}\otimes\mathcal{G}_{m}}^{(sa)},$ we have the
following inclusion%
\begin{align}
& \left\{ \left( W_{cov}\right) _{red}\in\mathcal{T}_{\mathcal{G}%
_{k_{1}}\otimes\mathcal{\cdots}\otimes\mathcal{G}_{k_{j}}}\mid W_{cov}\in 
\mathcal{T}_{\mathcal{G}_{1}\otimes\mathcal{\cdots}\otimes\mathcal{G}%
_{m}}\right\}  \tag{A30}  \label{A30} \\
& \subseteq\left\{ \left( W_{red}\right) _{cov}\in\mathcal{T}_{\mathcal{G}%
_{k_{1}}\otimes\mathcal{\cdots}\otimes\mathcal{G}_{k_{j}}}\right\} .  \notag
\end{align}
In view of definition (\ref{15}) of the covering norm and relation $\mathrm{%
tr}\left[ (W_{cov})_{red}\right] =\mathrm{tr}[W_{cov}],$ inclusion (\ref{A30}%
) implies%
\begin{align}
\left\Vert W_{red}\right\Vert _{cov} & =\inf_{\left( W_{red}\right) _{cov}\in%
\mathcal{T}_{\mathcal{G}_{k_{1}}\otimes\mathcal{\cdots}\otimes\mathcal{G}%
_{k_{j}}}}\mathrm{tr}\left[ (W_{red})_{cov}\right]  \tag{A31}  \label{A31} \\
& \leq\inf_{W_{cov}\in\mathcal{T}_{\mathcal{G}_{1}\otimes\mathcal{\cdots }%
\otimes\mathcal{G}_{m}}}\mathrm{tr}\left[ (W_{cov})_{red}\right]  \notag \\
& =\inf_{W_{cov}\in\mathcal{T}_{\mathcal{G}_{1}\otimes\mathcal{\cdots}\otimes%
\mathcal{G}_{m}}}\mathrm{tr}[W_{cov}]=\left\Vert W\right\Vert _{cov}  \notag
\end{align}
for each $W_{red}$ reduced from an operator $W\in\mathcal{T}_{\mathcal{G}%
_{1}\otimes\mathcal{\cdots}\otimes\mathcal{G}_{m}}^{(sa)}.$ This proves
property (3).

\section{proofs for section 5}

In this appendix we prove lemma 2 and find also some upper bounds needed for
our presentation in section 5.

\subsection{Proof of lemma 2}

For short, we further omit the below indices $\rho ,$ $\mathrm{M}_{S,\Lambda
}$ at notation $\mathcal{E}_{\rho ,\mathrm{M}_{S,\Lambda }}$ and denote by $%
\mathcal{M}_{\mathcal{E}}:=\{\mu \}$ the set of all normalized bounded
real-valued measures $\mu $, each returning all distributions $%
P_{(s_{1},...,s_{N})}^{(\mathcal{E})}$ of scenario $\mathcal{E}$\ as the
corresponding marginals and defined on the direct product space $(\Lambda
_{1}^{S_{1}}\times \cdots \times \Lambda _{N}^{S_{N}},$\textbf{\ }$\mathcal{F%
}_{\Lambda _{1}}^{\otimes S_{1}}\otimes \cdots \otimes \mathcal{F}_{\Lambda
_{N}}^{\otimes S_{N}}):=(\Lambda ^{\prime },\mathcal{F}^{\prime })$.

If scenario $\mathcal{E}$ admits an LHV model, then, by corollary 1, in set $%
\mathcal{M}_{\mathcal{E}},$ there is a probability measure. Since the total
variation norm of any probability measure is equal to $1$, parameter $\gamma
_{\mathcal{E}}=\inf_{\mu \in \mathcal{M}_{\mathcal{E}}}||\mu ||_{var}$ $=1.$

Conversely, let $\gamma _{\mathcal{E}}=1.$ As it is shown in section 2 of
this appendix, for a quantum scenario, set $\mathcal{M}_{\mathcal{E}}$
contains more than one element, so that, in view of convexity of $\mathcal{M}%
_{\mathcal{E}},$ this set is infinite. From relation 
\begin{equation}
\gamma _{\mathcal{E}}=\inf \{||\mu ||_{var}\mid \mu \in \mathcal{M}_{%
\mathcal{E}}\}=1  \tag{B1}  \label{B1}
\end{equation}%
it follows that, in set $\mathcal{M}_{\mathcal{E}},$ there is a sequence $%
\{\mu _{m}\},$ for which $||\mu _{m}||_{var}\rightarrow 1$ as $m\rightarrow
\infty $ and which is bounded in norm $||\cdot ||_{var}$. Note that,
equipped with the total variation norm $\left\Vert \cdot \right\Vert _{var},$
the linear space $\mathfrak{F}_{(\Lambda ^{\prime },\mathcal{F}^{\prime })}$
of all bounded real-valued measures on the measurable space $(\Lambda
^{\prime },\mathcal{F}^{\prime })$ is Banach. Therefore, there exists a
subsequence $\{\mu _{k_{m}}\}\subseteq \{\mu _{m}\}$ and a measure $%
\widetilde{\mu }\in \mathfrak{F}_{(\Lambda ^{\prime },\mathcal{F}^{\prime
})} $ such that%
\begin{equation}
\lim_{m\rightarrow \infty }\int f(\lambda ^{\prime })\mu _{k_{m}}(\mathrm{d}%
\lambda ^{\prime })=\int f(\lambda ^{\prime })\widetilde{\mu }(\mathrm{d}%
\lambda ^{\prime }),\text{ \ \ }\left\Vert \widetilde{\mu }\right\Vert
_{var}\leq 1,  \tag{B2}  \label{B2}
\end{equation}%
for all Borel measurable bounded real-valued functions $f$ on $(\Lambda
^{\prime },\mathcal{F}^{\prime }).$ Denote by $\widetilde{P}%
_{(s_{1},...,s_{N})}(\mathrm{d}\lambda _{1}^{(s_{1})}\times \cdots \times 
\mathrm{d}\lambda _{N}^{(s_{N})})$ the corresponding marginal of measure $%
\widetilde{\mu }(\mathrm{d}\lambda _{1}^{(1)}\times \cdots \times \mathrm{d}%
\lambda _{1}^{(S_{1})}\times \cdots \times $ $\mathrm{d}\lambda
_{N}^{(1)}\times \cdots \times \mathrm{d}\lambda _{N}^{(S_{N})}).$
Specifying (\ref{B2}) with the indicator function $\chi _{_{F}}((\lambda
_{1}^{(s_{1})},...,\lambda _{N}^{(s_{N})})),$ $F\in \mathcal{F}_{1}\otimes
\cdots \otimes \mathcal{F}_{N},$ and taking into the account that $\mu
_{k_{m}}\in \mathcal{M}_{\mathcal{E}},$ $\forall k_{m},$ we have 
\begin{equation}
P_{(s_{1},...,s_{N})}^{(\mathcal{E})}(F)=\widetilde{P}_{(s_{1},...,s_{N})}(F)
\tag{B3}  \label{B3}
\end{equation}%
for all sets $F\in \mathcal{F}_{1}\otimes \cdots \otimes \mathcal{F}_{N}$
and all tuples $(s_{1},...,s_{N}).$ Thus, $\widetilde{P}%
_{(s_{1},...,s_{N})}=P_{(s_{1},...,s_{N})}^{(\mathcal{E})}$ for all joint
measurements $(s_{1},...,s_{N})$ of scenario $\mathcal{E}.$ The latter means
that the bounded real-valued measure $\widetilde{\mu }$ is normalized and
returns all probability distributions $P_{(s_{1},...,s_{N})}^{(\mathcal{E})}$
of scenario $\mathcal{E}$ as the corresponding marginals. Hence, $\widetilde{%
\mu }$ belongs to set $\mathcal{M}_{\mathcal{E}}.$ In view of relation (\ref%
{35}) and the second relation in (\ref{B2}), the normalized bounded
real-valued measure $\widetilde{\mu }\in \mathcal{M}_{\mathcal{E}}$ is a
probability measure. By corollary 1, this proves lemma 2.

\subsection{Some upper bounds}

For the evaluation in proposition 4 of the state parameters (\ref{s1}), (\ref%
{s}), in addition to measure (\ref{22}), let us also consider another
example of possible measures $\mu_{\mathcal{E}_{\rho,\mathrm{M}%
_{S,\Lambda}}} $ in (\ref{q1}).

Let $T_{1\times S_{2}\times\cdots\times S_{N}}^{(\rho)}$ be an $1\times
S_{2}\times\cdots\times S_{N}$-setting source operator for state $\rho$ on $%
\mathcal{H}_{1}\otimes\cdots\otimes\mathcal{H}_{N}$ (see definition 1 in
section 2). Introduce the following collection of normalized bounded
real-valued measures%
\begin{align}
& \mathrm{tr}[T_{1\times S_{2}\times\cdots\times S_{N}}^{(\rho)}\text{ }\{%
\mathrm{M}_{1}^{(s_{1})}(\mathrm{d}\lambda_{1}^{(s_{_{1}})})\otimes \mathrm{M%
}_{2}^{(1)}(\mathrm{d}\lambda_{2}^{(1)})\otimes\cdots\otimes \mathrm{M}%
_{2}^{(S_{2})}(\mathrm{d}\lambda_{2}^{(S_{2})})  \tag{B4}  \label{B4} \\
& \otimes\cdots\otimes\mathrm{M}_{N}^{(1)}(\mathrm{d}\lambda_{N}^{(1)})%
\otimes\cdots\otimes\mathrm{M}_{N}^{(S_{N})}(\mathrm{d}\lambda
_{N}^{(S_{N})})\}],\text{ \ \ \ \ \ }s_{1}=1,...,S_{1},  \notag
\end{align}
where each $s_{1}$-th measure returns as the corresponding marginals all
joint distributions $P_{(s_{1},...,s_{N})}^{(\mathcal{E}_{\rho,\mathrm{M}%
_{S,\Lambda}})}$ of scenario $\mathcal{E}_{\rho,\mathrm{M}_{S,\Lambda}}$
with measurement $s_{1}$ at site $n=1.$

For an arbitrary trace class covering $(T_{1\times S_{2}\times \cdots \times
S_{N}}^{(\rho )}$ $)_{cov}$ of a source operator $T_{1\times S_{2}\times
\cdots \times S_{N}}^{(\rho )}$ (see definition 3 in section 2),
decomposition (\ref{11}) implies the following representation 
\begin{equation}
\frac{1}{2}\mathrm{tr}\left[ \mathfrak{\tau }^{+}\{\text{ }\mathrm{M}%
_{1}^{(s_{1})}(\cdot )\otimes \mathrm{M}_{2}^{(1)}(\cdot )\otimes \cdots \}%
\right] -\frac{1}{2}\mathrm{tr}\left[ \tau ^{-}\{\text{ }\mathrm{M}%
_{1}^{(s_{1})}(\cdot )\otimes \mathrm{M}_{2}^{(1)}(\cdot )\otimes \cdots \}%
\right]  \tag{B5}  \label{B5}
\end{equation}%
of each $s_{1}$-th measure (\ref{B4}) via two positive real-valued measures,
where 
\begin{equation}
\tau ^{\pm }=\left( T_{1\times S_{2}\times \cdots \times S_{N}}^{(\rho
)}\right) _{_{cov}}\pm T_{1\times S_{2}\times \cdots \times S_{N}}^{(\rho )}%
\overset{\otimes }{\geq 0}  \tag{B6}  \label{B6}
\end{equation}%
are tensor positive trace class operators.

As it is discussed in the proof of theorem 2 in Ref. \cite{loubenets3}, for
a positive measure $\nu $ on a direct product space $(\Lambda _{1},\mathcal{F%
}_{\Lambda _{1}})\times (\Lambda _{2},\mathcal{F}_{\Lambda _{2}})$, each
positive measure $\nu (B_{1}\times \cdot ),$ $B_{1}\in \mathcal{F}_{\Lambda
_{1}},$ on $(\Lambda _{2},\mathcal{F}_{\Lambda _{2}})$ is absolutely
continuous \cite{35} with respect to the marginal measure $\nu (\Lambda
_{1}\times \cdot )$ on $(\Lambda _{2},\mathcal{F}_{\Lambda _{2}})$. Hence,
for each of positive measures in decomposition (\ref{B5}), the Radon-Nykodim
theorem \cite{dunford} implies representation 
\begin{align}
& \mathrm{tr}\left[ \tau ^{\pm }\{\mathrm{M}_{1}^{(s_{1})}(\mathrm{d}\lambda
_{1}^{(s_{_{1}})})\otimes \mathrm{M}_{2}^{(1)}(\mathrm{d}\lambda
_{2}^{(1)})\otimes \cdots \}\right]  \tag{B7}  \label{B7} \\
& =\alpha _{s_{1}}^{(\pm )}(\mathrm{d}\lambda _{1}^{(s_{_{1}})}\mid \lambda
_{2}^{(1)},...)\text{ }\mathrm{tr}\left[ \tau ^{\pm }\{\mathbb{I}_{\mathcal{H%
}_{1}}\otimes \mathrm{M}_{2}^{(1)}(\mathrm{d}\lambda _{2}^{(1)})\otimes
\cdots \}\right] ]  \notag
\end{align}%
via conditional probability measure $\alpha _{s_{1}}^{(\pm )}(\cdot \mid
\lambda _{2}^{(1)},...)$ and marginal 
\begin{equation}
\mathrm{tr}\left[ \tau ^{\pm }\{\mathbb{I}_{\mathcal{H}_{1}}\otimes \mathrm{M%
}_{2}^{(1)}(\mathrm{d}\lambda _{2}^{(1)})\otimes \cdots \}\right] .  \tag{B8}
\label{B8}
\end{equation}%
From (\ref{B4}) - (\ref{B7}) it follows that the normalized bounded
real-valued measure%
\begin{align}
& \mu _{T_{1\times S_{2}\times \cdots \times S_{N}}^{(\rho )},\left(
T_{1\times S_{2}\times \cdots \times S_{N}}^{(\rho )}\right) _{cov}}^{(\rho ,%
\mathrm{M}_{S,\Lambda })}\left( \mathrm{d}\lambda _{1}^{(1)}\times \cdots
\times \mathrm{d}\lambda _{1}^{(S_{1})}\times \cdots \times \mathrm{d}%
\lambda _{N}^{(1)}\times \cdots \times \mathrm{d}\lambda
_{N}^{(S_{N})}\right)  \tag{B9}  \label{B9} \\
& :=\frac{1}{2}\left\{ \dprod\limits_{s_{1}=1,..,S_{1}}\alpha _{s_{1}}^{(+)}(%
\mathrm{d}\lambda _{1}^{(s_{_{1}})}|\lambda _{2}^{(1)},...)\right\} \text{ }%
\mathrm{tr}\left[ \tau ^{^{+}}\left\{ \mathbb{I}_{\mathcal{H}_{1}}\otimes 
\mathrm{M}_{2}^{(1)}(\mathrm{d}\lambda _{2}^{(1)})\otimes \cdots \right\} %
\right]  \notag \\
& -\frac{1}{2}\left\{ \dprod\limits_{s_{1}=1,..,S_{1}}\alpha _{s_{1}}^{(-)}(%
\mathrm{d}\lambda _{1}^{(s_{_{1}})}|\lambda _{2}^{(1)},...)\right\} \text{ }%
\mathrm{tr}\left[ \tau ^{-}\left\{ \text{ }\mathbb{I}_{\mathcal{H}%
_{1}}\otimes \mathrm{M}_{2}^{(1)}(\mathrm{d}\lambda _{2}^{(1)})\otimes
\cdots \right\} \right]  \notag
\end{align}%
\medskip returns all distributions (\ref{3}) of scenario $\mathcal{E}_{\rho ,%
\mathrm{M}_{S,\Lambda }}$ as the corresponding marginals.

\begin{lemma}
For arbitrary source operators $T_{S_{1}\times\cdots\times S_{N}}^{(\rho)},$ 
$T_{_{1\times S_{2}\times\cdots\times S_{N}}}^{(\rho)}$ for state $\rho,$
the total variation norms of measures (\ref{22}), (\ref{B9}) satisfy
relations%
\begin{equation}
1\leq\left\Vert \mu_{T_{S_{1}\times\cdots\times S_{N}}^{(\rho)}}^{(\rho,%
\mathrm{M}_{S,\Lambda})}\right\Vert _{_{_{var}}}\leq\left\Vert
T_{S_{1}\times\cdots\times S_{N}}^{(\rho)}\right\Vert _{cov}  \tag{B10}
\label{B10}
\end{equation}
and 
\begin{align}
1 & \leq\inf_{\left( T_{_{1\times S_{2}\times\cdots\times S_{N}}}^{(\rho
)}\right) _{_{_{cov}}}}\left\Vert \mu_{T_{_{1\times S_{2}\times\cdots\times
S_{N}}}^{(\rho)},\text{ }(T_{_{1\times S_{2}\times\cdots\times
S_{N}}}^{(\rho)}\text{ })_{_{_{cov}}}}^{(\rho,\mathrm{M}_{S,\Lambda})}\right%
\Vert _{_{var}}\leq\left\Vert T_{_{1\times S_{2}\times\cdots\times
S_{N}}}^{(\rho )}\right\Vert _{cov}  \tag{B11}  \label{B11} \\
&  \notag
\end{align}
for every collection $\mathrm{M}_{S,\Lambda}$ of POV measures on $\Lambda$
and an arbitrary outcome set $\Lambda.$ Here, infimum is taken over all
trace class coverings $(T_{_{1\times S_{2}\times\cdots\times S_{N}}}^{(\rho
)})_{_{_{cov}}}$ of a source operator $T_{_{1\times S_{2}\times\cdots\times
S_{N}}}^{(\rho)}$ and $\left\Vert \cdot\right\Vert _{cov}$ is the covering
norm (see definition 4 in section 2).
\end{lemma}

\begin{proof}
In view of (\ref{11}), consider the decomposition of a source operator 
\begin{equation}
T_{S_{1}\times ...\times S_{N}}^{(\rho )}=\frac{1}{2}\left\{ (T_{S_{1}\times
\cdots \times S_{N}}^{(\rho )})_{cov}+T_{S_{1}\times \cdots \times
S_{N}}^{(\rho )}\right\} -\frac{1}{2}\left\{ (T_{S_{1}\times \cdots \times
S_{N}}^{(\rho )})_{cov}-T_{S_{1}\times \cdots \times S_{N}}^{(\rho )}\right\}
\tag{B12}  \label{B12}
\end{equation}%
via tensor positive operators $(T_{S_{1}\times \cdots \times S_{N}}^{(\rho
)})_{cov}\pm T_{S_{1}\times \cdots \times S_{N}}^{(\rho )},$ where $%
(T_{S_{1}\times \cdots \times S_{N}}^{(\rho )})_{_{cov}}$ is an arbitrary
trace class covering of $T_{S_{1}\times \cdots \times S_{N}}^{(\rho )},$ see
definition 3 of section 2.

Substituting (\ref{B12}) into (\ref{22}), we represent measure $\mu
_{T_{S_{1}\times \cdots \times S_{N}}^{(\rho )}}^{(\rho ,\mathrm{M}%
_{S,\Lambda })}$ as the difference of two positive measures and, for each of
these measures, the total variation (\ref{tv}) is upper bounded by 
\begin{equation}
\frac{1}{2}\mathrm{tr}\left[ \left( T_{S_{1}\times \cdots \times
S_{N}}^{(\rho )}\right) _{cov}\pm T_{S_{1}\times \cdots \times S_{N}}^{(\rho
)}\right] .  \tag{B13}  \label{B13}
\end{equation}%
\medskip Therefore, for the total variation (\ref{tv}) of the normalized
measure $\mu _{T_{S_{1}\times \cdots \times S_{N}}^{(\rho )}}^{(\rho ,%
\mathrm{M}_{S,\Lambda })},$ we have%
\begin{equation}
1\leq \left\Vert \mu _{T_{S_{1}\times \cdots \times S_{N}}^{(\rho )}}^{(\rho
,\mathrm{M}_{S,\Lambda })}\right\Vert _{var}\leq \mathrm{tr}\left[ \left(
T_{S_{1}\times \cdots \times S_{N}}^{(\rho )}\right) _{cov}\right]  \tag{B14}
\label{B14}
\end{equation}%
\medskip for every trace class covering $(T_{S_{1}\times \cdots \times
S_{N}}^{(\rho )})_{_{cov}}$ of a source operator $T_{S_{1}\times \cdots
\times S_{N}}^{(\rho )}.$ Relations (\ref{B14}), (\ref{15}) imply bound (\ref%
{B10}).

Quite similarly, for measure (\ref{B9}), 
\begin{subequations}
\begin{equation}
1\leq \left\Vert \mu _{T_{_{1\times S_{2}\times \cdots \times S_{N}}}^{(\rho
)},\left( T_{_{1\times S_{2}\times \cdots \times S_{N}}}^{(\rho )}\right)
_{_{_{cov}}}}^{(\rho ,\mathrm{M}_{S,\Lambda })}\right\Vert _{var}\leq 
\mathrm{tr}\left[ \left( T_{_{1\times S_{2}\times \cdots \times
S_{N}}}^{(\rho )}\text{ }\right) _{cov}\right]  \tag{B15}  \label{B15}
\end{equation}%
\medskip for each trace class covering $(T_{_{1\times S_{2}\times \cdots
\times S_{N}}}^{(\rho )})_{cov}$ of a source operator $T_{_{1\times
S_{2}\times \cdots \times S_{N}}}^{(\rho )}$. From (\ref{B15}), (\ref{15})
we derive 
\end{subequations}
\begin{equation}
1\leq \inf_{_{\left( T_{_{1\times S_{2}\times \cdots \times S_{N}}}^{(\rho
)}\right) _{_{cov}}}}\text{ }\left\Vert \mu _{T_{_{1\times S_{2}\times
\cdots \times S_{N}}}^{(\rho )},\left( T_{_{1\times S_{2}\times \cdots
\times S_{N}}}^{(\rho )}\right) _{_{cov}}}^{(\rho ,\mathrm{M}_{S,\Lambda
})}\right\Vert _{var}\leq \left\Vert T_{_{1\times S_{2}\times \cdots \times
S_{N}}}^{(\rho )}\right\Vert _{cov},  \tag{B16}  \label{B16}
\end{equation}%
representing relation (\ref{B11}).
\end{proof}

Generalizing measure (\ref{B9}) to the case of a source operator $%
T_{S_{1}\times\cdots\times\underset{\overset{\uparrow}{n}}{1}\times
\cdots\times S_{N}}^{(\rho)}$ with $S_{n}=1$ at an arbitrary $n$-th site,
similarly to bound (\ref{B11}), we have 
\begin{align}
1 & \leq\inf_{(T_{S_{1}\times\cdots\times\underset{\overset{\uparrow}{n}}{1}%
\times\cdots\times S_{N}}^{(\rho)})_{_{_{cov}}}}\text{ }||\mu
_{T_{S_{1}\times\cdots\times\underset{\overset{\uparrow}{n}}{1}\times
\cdots\times S_{N}}^{(\rho)},(T_{S_{1}\times\cdots\times\underset{\overset{%
\uparrow}{n}}{1}\times\cdots\times S_{N}}^{(\rho)})_{cov}}^{(\rho,\mathrm{M}%
_{S,\Lambda})}\text{ }||_{_{var}}  \tag{B17}  \label{B17} \\
& \leq||T_{S_{1}\times\cdots\times\underset{\overset{\uparrow}{n}}{1}%
\times\cdots\times S_{N}}^{(\rho)}\text{ }||_{_{cov}}.  \notag
\end{align}
\medskip

Bounds (\ref{B13}), (\ref{B17}) allow us to evaluate the scenario parameter $%
\gamma_{\mathcal{E}_{\rho,\mathrm{M}_{S,\Lambda}}}$ defined by relation (\ref%
{q1}).

\begin{lemma}
For each $S_{1}\times\cdots\times S_{N}$-setting correlation scenario $%
\mathcal{E}_{\rho,\mathrm{M}_{S,\Lambda}}$ on a state $\rho$ on a Hilbert
space $\mathcal{H}_{1}\otimes\mathcal{\cdots}\otimes\mathcal{H}_{N}$, 
\begin{equation}
\gamma_{\mathcal{E}_{\rho,\mathrm{M}_{S,\Lambda}}}\leq\inf_{T_{S_{1}\times%
\cdots\times\underset{\overset{\uparrow}{n}}{1}\times\cdots\times
S_{N}}^{(\rho)},\text{ }\forall n}\text{ }||T_{S_{1}\times\cdots\times 
\underset{\overset{\uparrow}{n}}{1}\times\cdots\times S_{N}}^{(\rho)}\text{ }%
||_{_{cov}}.  \tag{B18}  \label{B18}
\end{equation}
\end{lemma}

\begin{proof}
From \ref{q1}), (\ref{B10}), (\ref{B17}), it follows%
\begin{equation}
\gamma _{\mathcal{E}_{\rho ,\mathrm{M}_{S,\Lambda }}}\leq \inf \left\{ 
\mathbf{||}T_{S_{1}\times \cdots \times \underset{\overset{\uparrow }{n}}{1}%
\times \cdots \times S_{N}}^{(\rho )}\text{ }\mathbf{||}_{cov},\text{ }%
\left\Vert T_{S_{1}\times \cdots \times S_{N}}^{(\rho )}\right\Vert _{cov},%
\text{ }n=1,...,N\right\} ,  \tag{B19}  \label{B19}
\end{equation}%
where infimum is taken over all source operators $T_{S_{1}\times \cdots
\times S_{N}}^{(\rho )},$ $T_{S_{1}\times \cdots \times \underset{\overset{%
\uparrow }{n}}{1}\times \cdots \times S_{N}}^{(\rho )}$ and\ over all $%
n=1,...,N.$ For each $n,$ the set $\{T_{S_{1}\times \cdots \times \underset{%
\overset{\uparrow }{n}}{1}\times \cdots \times S_{N}}^{(\rho )}\}$ of all
source operators $T_{S_{1}\times \cdots \times \underset{\overset{\uparrow }{%
n}}{1}\times \cdots \times S_{N}}^{(\rho )}$ for state $\rho $ includes the
set of all source operators on $\mathcal{H}_{1}^{(\otimes S_{1})}\otimes 
\mathcal{\cdots \otimes H}_{n}\otimes \cdots \otimes \mathcal{H}%
_{N}^{(S_{N})}$, each reduced from some $T_{S_{1}\times \cdots \times
S_{N}}^{(\rho )},$ as a particular subset. This inclusion and relation (\ref%
{20}) imply%
\begin{equation}
\inf_{T_{S_{1}\times \cdots \times \underset{\overset{\uparrow }{n}}{1}%
\times \cdots \times S_{N}}^{(\rho )}}\text{ }||T_{S_{1}\times \cdots \times 
\underset{\overset{\uparrow }{n}}{1}\times \cdots \times S_{N}}^{(\rho
)}||_{cov}\text{ }\leq \text{\ }\inf_{T_{_{S_{1}\times \cdots \times
S_{N}}}^{(\rho )}}\left\Vert T_{S_{1}\times \cdots \times S_{N}}^{(\rho
)}\right\Vert _{cov}  \tag{B20}  \label{B20}
\end{equation}%
for each $n$. Taking this into the account in (\ref{B19}), we prove (\ref%
{B18}).
\end{proof}

\section{proof of lemma 3 in section 6}

Constraints (\ref{52}) imply 
\begin{equation}
\sup_{_{\Psi _{S,\Lambda },\mathcal{B}_{\Psi _{S,\Lambda }}\neq
0}}\left\vert \frac{1}{\mathcal{B}_{\Psi _{S,\Lambda }}}\text{ }%
\dsum\limits_{s_{1},...,s_{_{N}}}\left\langle \psi _{(s_{1},\ldots
,s_{_{N}})}(\lambda _{1},...,\lambda _{N})\right\rangle _{\mathcal{E}_{\rho ,%
\mathrm{M}_{S,\Lambda }}}\right\vert \leq \mathrm{\gamma }_{\mathcal{E}%
_{\rho ,\mathrm{M}_{S,\Lambda }}},  \tag{C1}  \label{C1}
\end{equation}%
where supremum is taken over all non-trivial (i. e. $\mathcal{B}_{\Psi
_{S,\Lambda }}\neq 0$) function collections $\Psi _{S,\Lambda }.$ For short
of notation, we further replace $\mathcal{E}_{\rho ,\mathrm{M}%
_{S}}\rightarrow \mathcal{E}_{S,\Lambda }$. In order to prove 
\begin{equation}
\sup_{_{\Psi _{S,\Lambda },\mathcal{B}_{\Psi _{S,\Lambda }}\neq
0}}\left\vert \frac{1}{\mathcal{B}_{\Psi _{S,\Lambda }}}\text{ }%
\dsum\limits_{s_{1},...,s_{_{N}}}\left\langle \psi _{(s_{1},\ldots
,s_{_{N}})}(\lambda _{1},...,\lambda _{N})\right\rangle _{\mathcal{E}%
_{S,\Lambda }}\right\vert =\mathrm{\gamma }_{\mathcal{E}_{S,\Lambda }}, 
\tag{C2}  \label{C2}
\end{equation}%
we note that, by introducing variables $\xi ^{\pm }=\mu _{\mathcal{E}%
_{S,\Lambda }}^{\pm }(\Omega )\geq 0,$ the parameter $\mathrm{\gamma }_{%
\mathcal{E}_{S,\Lambda }}$, defined by (\ref{q1}), can be otherwise
expressed as 
\begin{align}
\mathrm{\gamma }_{\mathcal{E}_{S,\Lambda }}& =\inf \{\xi ^{+}+\xi ^{-}\mid
\xi ^{\pm }\geq 0,\text{ \ }\xi ^{+}-\xi ^{-}=1,\text{ }\exists \text{ }%
\mathcal{E}_{S,\Lambda }^{lhv},\widetilde{\mathcal{E}}_{S,\Lambda }^{\text{ }%
lhv}:  \tag{C3}  \label{C3} \\
P_{(s_{1},...,s_{n})}^{(\mathcal{E}_{S,\Lambda })}& =\xi
^{+}P_{(s_{1},...,s_{n})}^{(\mathcal{E}_{S,\Lambda }^{\text{ }lhv})}-\xi
^{-}P_{(s_{1},...,s_{n})}^{(\widetilde{\mathcal{E}}_{S,\Lambda }^{\text{ }%
lhv})},\text{ \ \ }\forall s_{n},\forall n\}.  \notag
\end{align}

As it is specified in section 3, we consider only standard measurable
spaces. In this case, $(\Lambda _{n},\mathcal{F}_{\Lambda _{n}})$ is Borel
isomorphic to some measurable space $(\mathcal{X}_{n},\mathcal{B}_{\mathcal{X%
}_{n}}),$ where $\mathcal{X}_{n}\mathcal{\in B}_{\mathbb{R}}$ is a Borel
subset of $\mathbb{R}$ and $\mathcal{B}_{\mathcal{X}_{n}}:=\mathcal{B}_{%
\mathbb{R}}\cap \mathcal{X}_{n}$ is the trace on $\mathcal{X}_{n}$ of the
Borel $\sigma $-algebra $\mathcal{B}_{\mathbb{R}}$ on $\mathbb{R}.$ We have
two major cases.

\textrm{(a)} \emph{Discrete case}$\emph{.}$ Let, for a correlation scenario $%
\mathcal{E}_{S,\Lambda },$ each outcome set be finite: $\Lambda
_{n}=\{\lambda _{n}^{(k_{n})},$ $k_{n}=1,...,$ $K_{n}<\infty \}$. Then 
\begin{align}
& \dsum\limits_{s_{1},...,s_{_{N}}}\left\langle \psi _{(s_{1},\ldots
,s_{_{N}})}(\lambda _{1},...,\lambda _{N})\right\rangle _{\mathcal{E}%
_{S,\Lambda }}  \tag{C4}  \label{C4} \\
& =\sum_{s_{1},...,s_{_{N}},\text{ }k_{1},...,k_{N}}\beta _{(s_{1},\ldots
,s_{_{N}})}^{(k_{1},...,k_{N})}P_{(s_{1},...,s_{n})}^{(\mathcal{E}%
_{S,\Lambda })}\left( \{\lambda _{1}^{(k_{1})}\}\times \cdots \times
\{\lambda _{N}^{(k_{n})}\}\right) ,  \notag
\end{align}%
where $\beta _{(s_{1},\ldots ,s_{_{N}})}^{(k_{1},...,k_{N})}$ are real
numbers. Hence, in (\ref{C1}), supremum over $\Psi _{S,\Lambda }$ reduces to
supremum over families $\{\beta _{(s_{1},\ldots
,s_{_{N}})}^{(k_{1},...,k_{N})}\}$ of real numbers and equality (\ref{C2})
follows from (\ref{C1}) and (\ref{C3}) by the linear programming (LP)
duality. This proof is similar to the proof of theorem 17 in Ref. \cite%
{kaplan} for a bipartite case with a finite number of outcomes at each site.

Let now, for a correlation scenario $\mathcal{E}_{S,\Lambda },$ every
outcome set $\Lambda _{n}$ be inifinite but countable: $\Lambda
_{n}=\{\lambda _{n}^{(k_{n})},k_{n}=1,...,$ $K_{n},...\}.$ Consider
collections $\Psi _{S,\Lambda }^{(\beta _{K},K)},$ of bounded measurable
real-valued functions%
\begin{equation}
\psi _{(s_{_{1}},...,s_{_{N}})}^{(\beta _{K},K)}(\lambda _{1},...,\lambda
_{N}):=\sum_{k_{1},...,k_{N}}\beta
_{(s_{_{1}},...,s_{_{N}})}^{(k_{1},...,k_{N})}\chi _{\{\lambda
_{1}^{(k_{1})}\}}(\lambda _{1})\cdot \ldots \cdot \chi _{\{\lambda
_{N}^{(k_{N})}\}}(\lambda _{N}),  \tag{C5}  \label{C5}
\end{equation}%
specified via tuples $K:=(K_{1},...,K_{N})$ of positive integers $%
K_{1},...,K_{N}<\infty $ \ and families $\beta _{K}:=\{\beta _{(s_{1},\ldots
,s_{_{N}})}^{(k_{1},...,k_{N})}\in \mathbb{R},$ $s_{n}=1,...,S_{n},$ $%
k_{n}=1,...,K_{n}\}$ of real numbers. For each of these function
collections, the expression%
\begin{align}
& \dsum\limits_{s_{1},...,s_{_{N}}}\left\langle \psi _{(s_{1},\ldots
,s_{_{N}})}^{(\beta _{K},K)}(\lambda _{1},...,\lambda _{N})\right\rangle _{%
\mathcal{E}_{S,\Lambda }}  \tag{C6}  \label{C6} \\
& =\sum_{s_{1},...,s_{_{N}},\text{ }k_{1},...,k_{N}}\beta
_{(s_{_{1}},...,s_{_{N}})}^{(k_{1},...,k_{N})}P_{(s_{_{1}},...,s_{_{N}})}^{(%
\mathcal{E}_{S,\Lambda })}\left( \{\lambda _{1}^{(k_{1})}\}\times \cdot
\cdot \cdot \times \{\lambda _{N}^{(k_{N})}\}\right)  \notag
\end{align}%
is similar by its form to (\ref{C4}), so that relation (\ref{C1}) and the
proof in the above case imply%
\begin{align}
\mathrm{\gamma }_{\mathcal{E}_{S,\Lambda }}& \geq \sup_{_{\Psi _{S,\Lambda },%
\mathcal{B}_{\Psi _{S,\Lambda }}\neq 0}}\text{ }\left\vert \text{ }\frac{1}{%
\mathcal{B}_{\Psi _{S,\Lambda }}}\dsum\limits_{s_{1},...,s_{_{N}}}\left%
\langle \psi _{(s_{_{1}},\ldots ,s_{_{N}})}(\lambda _{1},...,\lambda
_{N})\right\rangle _{\mathcal{E}_{S,\Lambda }}\right\vert \text{ }  \tag{C7}
\label{C7} \\
& \geq \sup_{_{\beta _{K},\text{ }K}}\text{ }\left\vert \frac{1}{\mathcal{B}%
_{\Psi _{S,\text{ }\Lambda }^{(\beta _{K},\text{ }K)}}}\dsum%
\limits_{s_{1},...,s_{_{N}}}\left\langle \psi
_{(s_{_{1}},...,s_{_{N}})}^{(\beta _{K},K)}(\lambda _{1},...,\lambda
_{N})\right\rangle _{\mathcal{E}_{S,\Lambda }}\right\vert \text{ }  \notag \\
& =\sup_{_{K}}\mathrm{\gamma }_{\mathcal{E}_{S,\Lambda }}^{(K)},  \notag
\end{align}%
where, for each $K=(K_{1},...,K_{N}),$ 
\begin{align}
\gamma _{\mathcal{E}_{S,\Lambda }}^{(K)}& :=\inf \{\text{ }\xi ^{+}+\xi
^{-}\mid \xi ^{\pm }\geq 0,\text{ \ }\xi ^{+}-\xi ^{-}=1,\text{\ \ }\exists 
\mathcal{E}_{S,\Lambda }^{\text{ }lhv},\text{ }\widetilde{\mathcal{E}}%
_{S,\Lambda }^{\text{ }lhv}:  \tag{C8}  \label{C8} \\
& P_{(s_{_{1}},...,s_{_{N}})}^{(\mathcal{E}_{S,\Lambda })}\left( \{\alpha
_{k_{1}}\}\times \cdots \times \{\alpha _{k_{N}}\}\right) =\xi
^{+}P_{(s_{1},...,s_{n})}^{(\mathcal{E}_{S,\Lambda }^{\text{ }lhv})}\left(
\{\lambda _{n}^{(k_{1})}\}\times \cdot \cdot \cdot \times \{\lambda
_{N}^{(k_{N})}\}\right)  \notag \\
& -\xi ^{-}P_{(s_{1},...,s_{n})}^{(\widetilde{\mathcal{E}}_{S,\Lambda }^{%
\text{ }lhv})}\left( \{\lambda _{n}^{(k_{1})}\}\times \cdot \cdot \cdot
\times \{\lambda _{N}^{(k_{N})}\}\right) ,\text{ \ }k_{n}=1,...,K_{n},\text{ 
}s_{n}=1,...,S_{n}\}.  \notag
\end{align}%
From (\ref{C3}), (\ref{C8}) it follows that $\mathrm{\gamma }_{\mathcal{E}%
_{S,\Lambda }}^{(K)}\leq \mathrm{\gamma }_{\mathcal{E}_{S,\Lambda
}}^{(L)}\leq \mathrm{\gamma }_{\mathcal{E}_{S,\Lambda }}$, if $K_{1}\leq
L_{1},...,K_{N}\leq L_{N},$ and $\lim_{K_{1},...,K_{N}\rightarrow \infty }%
\mathrm{\gamma }_{\mathcal{E}_{S,\Lambda }}^{(K)}=\mathrm{\gamma }_{\mathcal{%
E}_{S,\Lambda }}.$ Hence, $\sup_{_{K}}\mathrm{\gamma }_{\mathcal{E}%
_{S,\Lambda }}^{(K)}=\mathrm{\gamma }_{\mathcal{E}_{S,\Lambda }}.$ Taking
this into the account in (\ref{C7}), we derive 
\begin{align}
\mathrm{\gamma }_{\mathcal{E}_{S,\Lambda }}& \geq \sup_{_{\Psi _{S,\Lambda },%
\mathcal{B}_{\Psi _{S,\Lambda }}\neq 0}}\text{ }\left\vert \frac{1}{\mathcal{%
B}_{\Psi _{S,\Lambda }}}\dsum\limits_{s_{1},...,s_{_{N}}}\left\langle \psi
_{(s_{_{1}},...,s_{_{N}})}(\lambda _{1},...,\lambda _{N})\right\rangle _{%
\mathcal{E}_{S,\Lambda }}\right\vert  \tag{C9}  \label{C9} \\
& \geq \mathrm{\gamma }_{\mathcal{E}_{S,\Lambda }}.  \notag
\end{align}%
This proves equality (\ref{C2}) if set $\Lambda _{0}$ is infinite and
countable.

\textrm{(b)} \emph{Continuous case. }Let, for a correlation scenario $%
\mathcal{E}_{S,\Lambda },$ each outcome set $\Lambda _{n}$ be infinite and
uncountable. For positive integers $K_{n}\geq 1,$ introduce partitions 
\begin{align}
\mathrm{D}_{K_{n}}& =\{D_{K_{n}}^{(k_{n})}\in \mathcal{F}_{\Lambda _{n}}\mid
D_{K_{n}}^{(k_{n})}\neq \varnothing ,\text{ \ \ }D_{K_{n}}^{(k_{n})}\cap
D_{K_{n}}^{(k_{n}^{\prime })}=\varnothing ,\text{ \ \ }\forall k_{n}\neq
k_{n}^{\prime },  \tag{C10}  \label{C10} \\
\cup _{k}D_{K_{n}}^{(k_{n})}& =\Lambda _{n},\text{ \ \ }k_{n}=1,..,2^{K_{n}}%
\}  \notag
\end{align}%
of each set $\Lambda _{n}$, such that 
\begin{equation}
D_{K_{n}+1}^{(2k_{n}-1)}\cup D_{K_{n}+1}^{(2k_{n})}=D_{K_{n}}^{(k_{n})},%
\text{ \ \ }k_{n}=1,...,2^{K_{n}},\text{ \ \ }\forall K_{n}\in \mathbb{N}. 
\tag{C11}  \label{C11}
\end{equation}%
For some partitions $\mathrm{D}_{K_{1}},...,\mathrm{D}_{K_{N},}$ of sets $%
\Lambda _{1},...,\Lambda _{N},$ respectively, consider collections $%
\widetilde{\Psi }_{S,\Lambda }^{(\beta _{K},\text{ }K)}$ of real-valued
measurable functions 
\begin{equation}
\widetilde{\psi }_{(s_{_{1}},\ldots ,s_{_{N}})}^{(\beta _{K},K)}(\lambda
_{1},...,\lambda _{N}):=\sum_{k_{1},...,k_{N}}\beta _{(s_{_{1}},\ldots
,s_{_{N}})}^{(k_{1},...,k_{N})}\chi _{D_{K_{1}}^{(k_{1})}}(\lambda
_{1})\cdot \ldots \cdot \chi _{D_{K_{N}}^{(k_{N})}}(\lambda _{N}),  \tag{C12}
\label{C12}
\end{equation}%
specified by tuples $K:=(K_{1},...,K_{N})$ and collections $\beta
_{K}:=\{\beta _{(s_{1},\ldots ,s_{_{N}})}^{(k_{1},...,k_{N})}\in \mathbb{R},$
$s_{n}=1,...,S_{n},$ $k_{n}=1,...,2^{K_{n}}\}$ of real numbers. Similarly to
the derivation of (\ref{C7}), we have:%
\begin{align}
\mathrm{\gamma }_{\mathcal{E}_{S,\Lambda }}& \geq \sup_{_{\Psi _{S,\Lambda },%
\mathcal{B}_{\Psi _{S,\Lambda }}\neq 0}}\left\vert \text{ }\frac{1}{\mathcal{%
B}_{\Psi _{S,\Lambda }}}\dsum\limits_{s_{1},...,s_{_{N}}}\left\langle \psi
_{(s_{_{1}},\ldots ,s_{_{N}})}(\lambda _{1},...,\lambda _{N})\right\rangle _{%
\mathcal{E}_{S,\Lambda }}\right\vert \text{ }  \tag{C13}  \label{C13} \\
& \geq \sup_{_{\beta _{K},K}}\text{ }\left\vert \frac{1}{\mathcal{B}_{%
\widetilde{\Psi }_{S,\Lambda }^{(\beta _{K},\text{ }K)}}}\dsum%
\limits_{s_{1},...,s_{_{N}}}\left\langle \widetilde{\psi }_{(s_{_{1}},\ldots
,s_{_{N}})}^{(\beta _{K},K)}(\lambda _{1},...,\lambda _{N})\right\rangle _{%
\mathcal{E}_{S,\Lambda }}\right\vert \text{ }  \notag \\
& =\sup_{_{K}}\mathrm{\gamma }_{\mathcal{E}_{S,\Lambda }}^{(K)},  \notag
\end{align}%
\medskip where%
\begin{align}
\mathrm{\gamma }_{\mathcal{E}_{S,\Lambda }}^{(K)}& :=\inf \{\text{ }\xi
^{+}+\xi ^{-}\mid \xi ^{\pm }\geq 0,\text{ \ }\xi ^{+}-\xi ^{-}=1,\text{\ \ }%
\exists \mathcal{E}_{S,\Lambda }^{\text{ }lhv},\widetilde{\mathcal{E}}%
_{S,\Lambda }^{\text{ }lhv}:  \tag{C14}  \label{C14} \\
& P_{(s_{1},...,s_{_{N}})}^{(\mathcal{E}_{S,\Lambda })}\left(
D_{K_{1}}^{(k_{1})}\times \cdots \times D_{K_{N}}^{(k_{N})}\right) =\xi
^{+}P_{(s_{1},...,s_{n})}^{(\mathcal{E}_{S,\Lambda }^{\text{ }lhv})}\left(
D_{K_{1}}^{(k_{1})}\times \cdots \times D_{K_{N}}^{(k_{N})}\right)  \notag \\
& -\xi ^{-}P_{(s_{1},...,s_{n})}^{(\widetilde{\mathcal{E}}_{S,\Lambda }^{%
\text{ }lhv})}\left( D_{K_{1}}^{(k_{1})}\times \cdots \times
D_{K_{N}}^{(k_{N})}\right) ,\text{\ \ \ }k_{n}=1,...,2^{K_{n}},\text{ \ }%
s_{n}=1,...,S_{n}\},  \notag
\end{align}%
for each $K=(K_{1},...K_{N}),$ $K_{n}\in \mathbb{N}.$ From (\ref{C3}), (\ref%
{C14}) and the special construction (\ref{C11}) of partitions $\mathrm{D}%
_{K},$ it follows that $\mathrm{\gamma }_{\mathcal{E}_{S,\Lambda
}}^{(K)}\leq \mathrm{\gamma }_{\mathcal{E}_{S,\Lambda }}^{(L)}\leq \mathrm{%
\gamma }_{\mathcal{E}_{S,\Lambda }}$, if $K_{1}\leq L_{1},...,K_{N}\leq
L_{N} $, and $\lim_{K_{1},...,K_{N}\rightarrow \infty }\mathrm{\gamma }_{%
\mathcal{E}_{S,\Lambda }}^{(K)}=\mathrm{\gamma }_{\mathcal{E}_{S,\Lambda }}.$
Therefore, $\sup_{K}\mathrm{\gamma }_{\mathcal{E}_{S,\Lambda }}^{(K)}=%
\mathrm{\gamma }_{\mathcal{E}_{S,\Lambda }}$. Substituting this into (\ref%
{C13}), we prove equality (\ref{C2}) in case of uncountable sets $\Lambda
_{n}$.

Coming back to notation $\mathcal{E}_{S,\Lambda }\rightarrow \mathcal{E}%
_{\rho ,\mathrm{M}_{S,\Lambda }}$ and taking supremum of the left-hand and
the right-hand sides of (\ref{C2}) over all collections $\mathrm{M}%
_{S,\Lambda }$ of POV measures, we prove relation (\ref{64}).


\begin{thebibliography}{99}
\bibitem{30} On this constant and Grothendieck's inequality see Refs. \cite%
{Finch, krivine}. A generalized Grothendieck's inequality has been recently
introduced in Ref. \cite{toner}.
\bibitem{32} We shall consider only complex separable Hilbert spaces and
henceforth suppress terms "complex" and "separable" with respect to a
Hilbert space.
\bibitem{33} Throughout this paper, for a measure $\mu$ on the direct
product of measurable spaces, we often use notation $\mu(\mathrm{d}%
\lambda\times \cdots\times\mathrm{d}\lambda^{\prime})$ outside of an
integral since this allows us: (i) to specify visually the structure of
different marginals of measure $\mu;$ (ii) to facilitate the description of
measures constructed via tensor products of POV measures and the computation
of their total variation norms.
\end{thebibliography}
\end{document}